\documentclass[sn-nature]{sn-jnl}% Style for submissions to Nature Portfolio journals
\usepackage{graphicx}%
\usepackage{multirow}%
\usepackage{amsmath,amssymb,amsfonts}%
\usepackage{amsthm}%
\usepackage{mathrsfs}%
\usepackage[title]{appendix}%
\usepackage{xcolor}%
\usepackage{textcomp}%
\usepackage{manyfoot}%
\usepackage{booktabs}%
\usepackage{algorithm}%
\usepackage{algorithmicx}%
\usepackage{algpseudocode}%
\usepackage{listings}%
\usepackage{threeparttable}
\usepackage{makecell}
\usepackage{subfigure}
\usepackage{bbding} %
\theoremstyle{thmstyleone}%
\newtheorem{theorem}{Theorem}%  meant for continuous numbers
\theoremstyle{thmstyletwo}%

\theoremstyle{thmstylethree}%
\newtheorem{definition}{Definition}%

\begin{document}

\title[Article Title]{Two-stage space construction for real-time modeling of distributed parameter systems under sparse sensing}
%\title[Article Title]{Sparse Sensing-Based Space Construction for Modeling Distributed Parameter Systems}
%\title[Article Title]{Two-stage space construction for modeling distributed parameter systems under sparse sensing}


\author[1,2,4]{\fnm{Peng} \sur{Wei}}\email{pengwei7-c@my.cityu.edu.hk}
%\author*[1,2,3]{\fnm{Changjun} \sur{Xie}}\email{jackxie@whut.edu.cn}

%\author*[1]{\fnm{Han-Xiong} \sur{Li}}\email{mehxli@cityu.edu.hk}

\affil[1]{\orgdiv{School of Automation}, \orgname{Wuhan University of Technology}, \orgaddress{\street{No.122 Luoshi Road}, \city{Wuhan}, \postcode{430070}, \state{Hubei}, \country{China}}}

\affil[2]{\orgdiv{Hubei Key Laboratory of Advanced Technology for Automotive Components}, \orgname{Wuhan University of Technology}, \orgaddress{\city{Wuhan}, \postcode{430070}, \country{China}}}

%\affil[3]{\orgdiv{Modern Industry College of Artificial Intelligence and New Energy Vehicles}, \orgname{Wuhan University of Technology}, \orgaddress{\city{Wuhan}, \postcode{430070}, \state{Hubei}, \country{China}}}

\affil[4]{\orgdiv{Department of Systems Engineering}, \orgname{City University of Hong Kong}, \orgaddress{\street{Tat Chee Avenue}, \city{Kowloon}, \postcode{999077}, \state{Hong Kong SAR}, \country{China}}}


%\abstract{Many industrial processes can be described by distributed parameter systems (DPSs). In this research, a two-stage space construction method is proposed for online modeling of DPSs under sparse sensing. First, the discrete space-completion scheme is developed to recover the spatiotemporal dynamics of non-sensing locations under sparse sensing. The continuous spatial basis functions (SBFs) are derived by the high-dimensional space construction method. The nonlinear temporal model is identified and updated by the long short-term memory (LSTM) neural network. Finally, the spatially continuous model is obtained by the synthesis of the derived SBFs and temporal model. The cubic B-spline surface is proven to be an appropriate solution for the optimization of space construction in the sense of least squares approximation. Experimental studies on a pouch-type Li-ion battery demonstrate the effectiveness of the proposed modeling method under sparse sensing. This work highlights the promise of sparse sensors in online full-space modeling in large-scale battery energy storage systems.}

\abstract{Numerous industrial processes can be defined using distributed parameter systems (DPSs). This study introduces a two-stage spatial construction approach for real-time modeling of DPSs in cases of limited sensors. Initially, a discrete space-completion approach is created to recuperate the spatiotemporal patterns of non-monitored locations under sparse sensing. The high-dimensional space construction method is employed to derive continuous spatial basis functions (SBFs). The identification and adjustment of the nonlinear temporal model are carried out via the long short-term memory (LSTM) neural network. Eventually, the amalgamation of the derived SBFs and temporal model results in a spatially continuous model. The use of a cubic B-spline surface is validated as an effective solution for optimizing space construction in the sense of least squares approximation. Experimental tests conducted on a pouch-type Li-ion battery demonstrate the efficacy of the proposed modeling technique under sparse sensing.  This work highlights the promise of sparse sensors in real-time full-space modeling for large-scale battery energy storage systems.}




\maketitle

\section{Introduction}\label{sec1}

%Distributed parameter systems (DPSs), also called spatiotemporal dynamics systems, can be used to describe many industrial processes, such as chemical catalytic processes \cite{christofides2002nonlinear}, battery thermal process \cite{wei2023spatio}, ultrasonic propagation processes \cite{vanhille2008nonlinear}. Accurate online models are the basis for fast fault diagnosis and real-time control \cite{dai2013model}. However, due to the spatiotemporal properties and nonlinear couplings between different spatial dimensions, it is challenging to derive an accurate online model for DPSs, especially in the case of sparse sensing.

Distributed parameter systems (DPSs), known as spatiotemporal dynamics systems, usually used to describe industrial processes, such as chemical catalytic processes \cite{christofides2002nonlinear}, battery thermal processes \cite{wei2023spatio}, and ultrasonic propagation processes \cite{vanhille2008nonlinear}. The development of precise online models serves as the foundation for rapid fault diagnosis and real-time control \cite{dai2013model}. However, the intricate spatiotemporal properties and complex nonlinear couplings among different spatial dimensions make it challenging to establish accurate online models for DPSs, particularly in situations involving sparse sensing.

%DPS modeling methods can be divided into first-principle methods and data-driven methods. First-principle methods, such as the finite element method \cite{reddy2019introduction}, finite difference method \cite{ozicsik2017finite,sattarzadeh2021real}, spectral method \cite{deng2005spectral}, utilize accurate system partial differential equations (PDEs) to derive a finite-order ODE model to approximate the original DPS. For example, a spectral-approximation-based reduced model was derived in \cite{deng2005spectral} for spatiotemporal modeling of two-dimensional DPSs. However, it is difficult to obtain accurate governing PDEs and corresponding boundary conditions in practice, which limits the application of these methods in industrial processes.

DPS modeling methods can be categorized into first-principle methods and data-driven methods. First-principle methods, such as the finite element method \cite{reddy2019introduction}, finite difference method \cite{ozicsik2017finite, sattarzadeh2021real}, and spectral method \cite{deng2005spectral}, employ accurate system partial differential equations (PDEs) to derive finite-order ODE models for approximating the original DPS. For instance, in the work by Deng \cite{deng2005spectral}, a spectral-approximation-based reduced model was developed for the spatiotemporal modeling of two-dimensional DPSs. However, the practical challenge lies in obtaining precise governing PDEs and their corresponding boundary conditions, which can limit the application of these methods in industrial processes.

%Data-driven methods, such as the Karhunen-Loeve (KL) method \cite{baker2000finite,chen2019dimension}, and its variations \cite{wang2018sliding,wang2018incremental}, utilize data measured by many sensors distributed in the spatial domain to model the DPS. For example, a sliding window-based method was proposed in \cite{wang2018sliding} for online modeling of DPSs. In order to reduce the computational complexity of online model, an incremental learning algorithm was proposed in \cite{wang2018incremental} to update spatial basis functions (SBFs) in a more efficient manner. These methods enable DPS modeling without relying on accurate system equations and are therefore more popular in practical applications. As a trade-off, data-based approaches require a large number of sensors for accurate modeling. However, it is desirable to arrange as few sensors as possible to reduce costs and system complexity in industrial applications. In order to solve this difficulty, a KL-based method was proposed in \cite{chen2020spatiotemporal} to model the DPS under sparse sensing. But this method still cannot achieve full-space modeling, because the measurement data is spatially discrete, which is also unavoidable in traditional data-based modeling methods. If the continuous SBFs can be derived, the full-space prediction may be achieved under sparse sensing.

Data-driven methods, such as the Karhunen-Loeve (KL) method \cite{baker2000finite, chen2019dimension}, and its variations \cite{wang2018sliding, wang2018incremental}, utilize data collected by multiple sensors distributed in the spatial domain to model the DPS. For instance, a sliding window-based method was proposed in \cite{wang2018sliding} for the online modeling of DPSs. Furthermore, to reduce the computational complexity of the online model, an incremental learning algorithm was introduced in \cite{wang2018incremental} to update spatial basis functions (SBFs) more efficiently. These approaches enable DPS modeling without the reliance on precise system equations, making them more prevalent in practical applications. However, data-based methods necessitate a substantial number of sensors for accurate modeling, posing challenges for cost and system complexity in industrial applications. Addressing this challenge, a KL-based method was suggested in \cite{chen2020spatiotemporal} to model DPS under sparse sensing. Nonetheless, this method still cannot achieve full-space modeling due to the spatially discrete nature of measurement data, which is inherent in traditional data-driven modeling techniques. Attaining continuous SBFs might enable the realization of full-space prediction even under sparse sensing conditions.

%For online modeling of DPSs, the temporal model usually needs to be continuously updated over time, which is time-consuming \cite{wang2018sliding}. On the one hand, the frequency of model updates should be reduced to improve the modeling efficiency. On the other hand, the frequency of model updates should be increased to improve the modeling accuracy. The long short-term memory (LSTM) neural network \cite{ojo2020neural,li2020novel}, a type of recurrent neural network, has strong nonlinear learning ability for sequence-type data. During the online prediction, the LSTM neural network can continuously update the model using the latest data without retraining \cite{zhang2022data}. Therefore, the LSTM neural may help to update the temporal model in DPS modeling.

When it comes to online modeling of DPSs, continuous updating of the temporal model often proves to be time-consuming \cite{wang2018sliding}. On one hand, minimizing the frequency of model updates is necessary to enhance modeling efficiency, while on the other hand, increasing the frequency of updates is essential for improving modeling accuracy. The long short-term memory (LSTM) neural network \cite{ojo2020neural, li2020novel}, a form of recurrent neural network, exhibits robust nonlinear learning capabilities for sequence-type data. During online prediction, the LSTM neural network can continuously update the model using the most recent data without necessitating retraining \cite{zhang2022data}. Consequently, the implementation of the LSTM neural network may facilitate the updating of the temporal model in DPS modeling.

In this work, we propose a two-stage space construction method for spatially continuous modeling of DPSs under sparse sensing. The proposed modeling framework comprises discrete space completion, continuous space construction, LSTM-based nonlinear learning, and space-time synthesis. Leveraging the suggested two-stage space construction approach enables the extraction of comprehensive space information even when dealing with sparse sensing. With the assistance of the LSTM neural network, the proposed spatiotemporal model can be continuously updated online. Furthermore, we analyze the influence of sensor quantity on modeling accuracy and examine the impact of various sensing schemes on modeling performance. The findings demonstrate that the proposed framework achieves full-space modeling with a training RMSE of 0.0457 and a testing RMSE of 0.0692 for pouch-type Li-ion batteries, employing only two online sensors. Comparative analysis reveals the superior performance of the proposed modeling method under identical sensing conditions. These outcomes underscore the potential of employing sparse sensors for full-space modeling of large-scale battery energy storage systems.

%In the future, more research work will focus on improving online performance and optimizing number and location of sensors. 

%The main contributions of this research are summarized as follows:
%\begin{itemize}
%	\item[1)] This study makes the first effort to design a systematic framework to enable full-space modeling of DPSs under sparse sensing.
%	\item[2)] Theoretical analysis proofs that the cubic B-spline surface is a feasible solution of the spatial construction optimization.
%	\item[3)] Experiment on a pouch-type Li-ion battery demonstrates the effectiveness and superiority of the proposed method. 
%\end{itemize}

\section{Results}\label{sec2}
\subsection{Framework overview}
%Li-ion batteries are widely used as power sources in electric vehicles and portable devices \cite{zhang2011review,xiong2020research}. The thermal process of pouch-type LIBs is a typical DPS \cite{wei2023spatio}. The battery cell as shown in Fig. \ref{fig:conception} is employed to validate the effectiveness of the proposed modeling method. As shown in Fig. \ref{fig:framework}, the proposed method mainly consists of the following four parts: discrete space completion, continuous space construction, LSTM-based nonlinear learning, and space-time synthesis for prediction. During the offline stage, the two-stage space construction is derived under full sensing. During the online stage, the full temporal coefficients can be obtained by iterative computation under sparse sensing. The complete space-time prediction can be achieved by the space-time synthesis of the spatially continuous SBFs $ \psi(x,y) $ and the corresponding temporal coefficients $ \hat{a}(t) $ derived by LSTM neural network.

Lithium-ion (Li-ion) batteries serve as prevalent power sources for electric vehicles and portable devices \cite{zhang2011review, xiong2020research}. The thermal process of pouch-type Li-ion batteries represents a typical distributed parameter system (DPS) \cite{wei2023spatio}. To validate the efficacy of the proposed modeling method, the battery cell depicted in Fig. \ref{fig:conception} is employed. The proposed approach, as illustrated in Fig. \ref{fig:framework}, primarily comprises discrete space completion, continuous space construction, LSTM-based nonlinear learning, and space-time synthesis for prediction. During the offline stage, the two-stage space construction is formulated under conditions of complete sensing. Subsequently, during the online stage, the entire temporal coefficients can be acquired through iterative computation under sparse sensing. The complete space-time prediction is achieved by synthesizing the spatially continuous SBFs $ \psi(x,y) $ and the corresponding temporal coefficients $ \hat{a}(t) $ derived using LSTM neural network.

\begin{figure}[htbp] 
	\centering
	\includegraphics[width=0.90\textwidth]{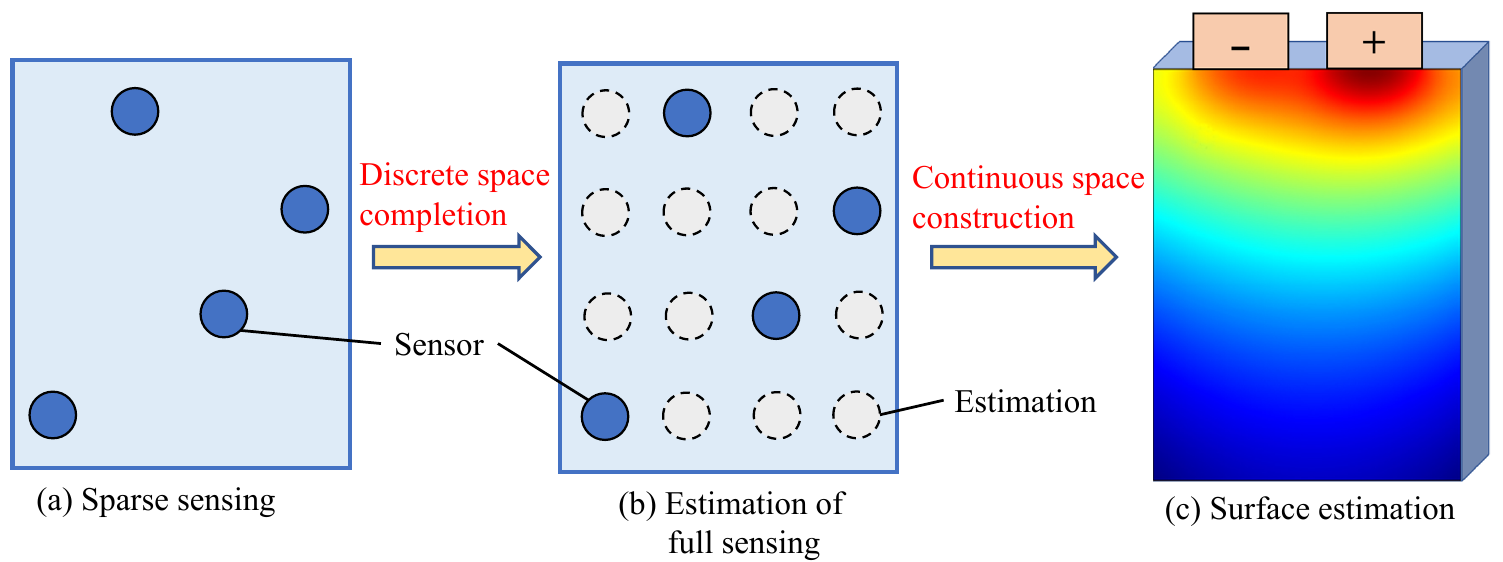}
	\caption{Concept illustration of the proposed discrete space completion and continuous space construction. (\textbf{a}) Sensor distribution under sparse sensing. (\textbf{b}) Temperature estimation of all sensors using the proposed discrete space completion method with dotted circles representing virtual sensors. (\textbf{c}) Temperature estimation of the full surface using the proposed continuous space construction method.}
	\label{fig:conception}	
\end{figure}

\begin{figure}[htbp] 
	\centering
	\includegraphics[width=0.65\textwidth]{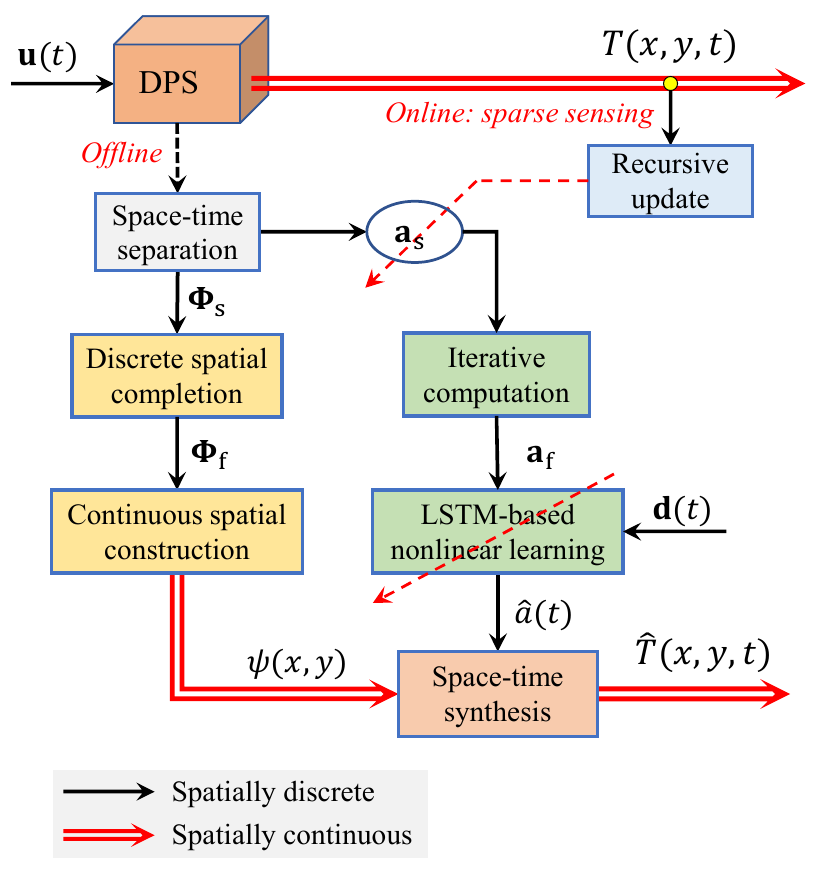}
	\caption{Framework of the proposed two-stage space construction modeling method. $ \mathbf{u}(t) $ denotes the system input vector of the DPS. $ T(x,y,t) $ and $ \widehat{T}(x,y,t) $ represent the actual spatiotemporal output and the predicted output of the DPS. During the offline stage, the collected sparse data $ \mathbf{T}_\text{s}$ are used to extract sparse spatial basis function (SBF) matrix $ \mathbf{\Phi}_\text{s} $ and corresponding temporal coefficient matrix $ \mathbf{a}_\text{s} $ by space-time separation. The full SBF matrix $ \mathbf{\Phi}_\text{f} $ is revealed by the proposed discrete spatial completion algorithm. At the same time, the spatially continuous SBF $ \psi(x,y) $ is derived by the proposed spatial construction method. During the online procedure, the sparse temporal coefficient matrix $ \mathbf{a}_\text{s}$ will be updated and delivered to the iterative computation module for the full temporal coefficient matrix $ \mathbf{a}_\text{f} $ derivation. The future temporal coefficient $ \widehat{a}(t) $ will be identified by the nonlinear learning algorithm. Finally, the spatially continuous modeling under sparse sensing can be achieved by the space-time synthesis.}
	\label{fig:framework}	
\end{figure}

\subsection{Data generation}
%The battery cell used in the experiment is a 0.15$ \times $0.20 pouch-type Li-ion battery. The nominal capacity, rated voltage, discharge cut off voltage, maximum charge current, maximum charge voltage of the battery cell are 20 Ah, 3.2 V, 2.0 V, 40.0 A, 3.65 V, respectively. The experiment testing platform is shown in Fig. \ref{fig:platform}. The battery cell is tested in the thermal chamber which is used to provide a stable test environment. The battery management system (BMS) is used to collect voltage, current and temperature data and transmit them to the host computer synchronously. The battery test system (BTS) is used to generate different test current waveform to the battery cell according to the control signal from the host computer. 

The experiment employs a 0.15$ \times $0.20 pouch-type Li-ion battery cell. The nominal capacity, rated voltage, discharge cut-off voltage, maximum charge current, and maximum charge voltage of the battery cell are 20 Ah, 3.2 V, 2.0 V, 40.0 A, and 3.65 V, respectively. The experimental testing platform is illustrated in Figure \ref{fig:platform}. The thermal chamber serves as the controlled testing environment for the battery cell. The battery management system (BMS) is responsible for the synchronous collection and transmission of voltage, current, and temperature data to the host computer. Additionally, the battery test system (BTS) is utilized to administer various test current waveforms to the battery cell, following the control signal from the host computer.

\begin{figure}[htbp] 
	\centering
	\includegraphics[width=0.65\textwidth]{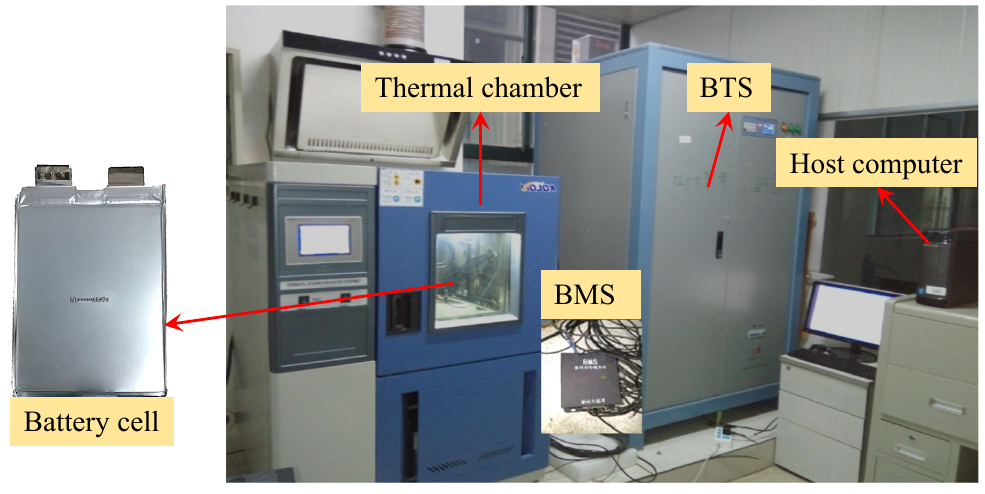}
	\caption{The experimental platform is comprised of several key components, including a thermal chamber, a battery management system (BMS), a battery test system (BTS), and a host computer. The battery under investigation is of the pouch-type cell, characterized by dimensions of 0.15 in width and 0.20 in length.}
	\label{fig:platform}	
\end{figure}

\begin{figure}[htbp] 
	\centering
	\includegraphics[width=0.65\textwidth]{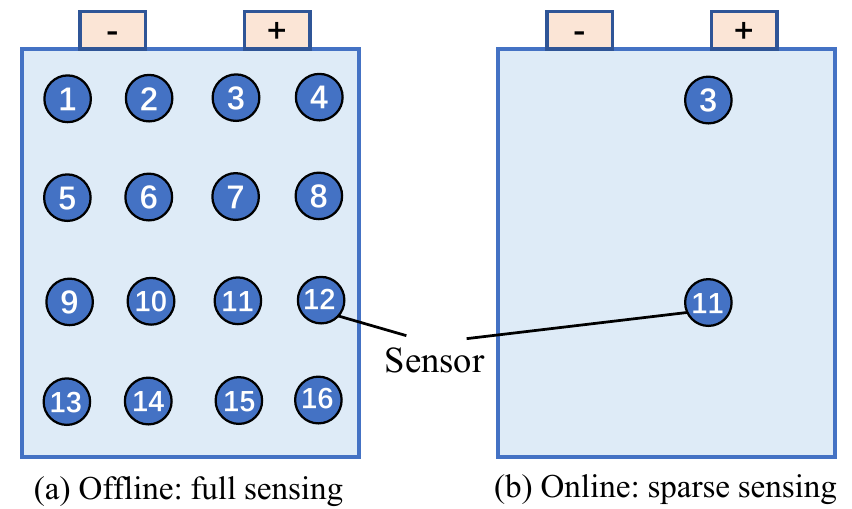}
	\caption{Sensor distribution in offline and online phases. (\textbf{a}) During the offline stage, the full sensing scheme is implemented, featuring a total of sixteen temperature sensors evenly distributed across the battery surface. (\textbf{b}) During the online procedure, the sparse sensing scheme is adopted with only two sensors deployed from the original sixteen ones.}
	\label{fig:sensing_diagram}	
\end{figure}

%As shown in Fig. \ref{fig:sensing_diagram}, sixteen temperature sensors are evenly distributed on the battery surface during the offline stage. The serial number is marked on the corresponding sensor position. During the online stage, only three of them are used for data collection and spatial construction. Temperature sensors are T-type thermalcouples. The environment temperature is stabilized at 25 $^{\circ}$C. The model order is determined by energy ratio method \cite{li2010modeling}. The main experiment parameters are listed in Table \ref{tab:parameters}. The load current in ladder form is employed to excite the battery cell as shown in Fig. \ref{fig:current}. 

As shown in Fig. \ref{fig:sensing_diagram}, sixteen temperature sensors are evenly distributed on the battery surface during the offline stage, with each sensor's position marked with a corresponding serial number. During the online stage, data collection and spatial construction utilize only three of these sensors. The temperature sensors utilized are T-type thermocouples, while the environmental temperature remains stabilized at 25 $^{\circ}$C. The determination of the model order is executed using the energy ratio method \cite{li2010modeling}. The primary experimental parameters are itemized in Table \ref{tab:parameters}. The load current follows a ladder form, serving as the excitation for the battery cell, as shown in Fig. \ref{fig:current}.

\begin{table}[htbp] 
	\centering
	\caption{Main Experimental Parameters.}
	\label{tab:parameters}
	\begin{tabular}{cccc}
		\toprule
		Parameter  &  Value & Unit & Source \\
		\midrule
		Length of the battery      &     0.200  & m & Measured\\
		Width of the battery    &     0.150 & m & Measured\\
		Thickness of the battery    &     7e-3  & m & Measured\\
		Nominal capacity           &   20    & Ah & Specification\\
		Rated voltage            &   3.2    & V & Specification\\
		Ambient temperature        &  25    & $^{\circ}$C & Selected\\
		Full sensor number $ N_\text{f} $    &      16       & - & Selected \\
		Sparse sensor number $ N_\text{s} $    &      3       & - & Selected \\
		Model order $ n $	       & 	 2   & -	& Derived \\			
		\bottomrule
	\end{tabular}
\end{table}

\begin{figure}[htbp] 
	\centering
	\includegraphics[width=0.6\textwidth]{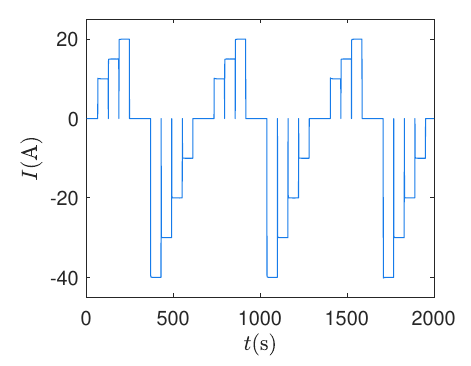}
	\caption{Load current profile from 0 to 2,000 s.}
	\label{fig:current}	
\end{figure}

\subsection{Experiment results}

%As shown in Fig. \ref{fig:sensing_diagram}, many sensors can be uniformly arranged in the space domain (full sensing scheme) during the offline training, while we should reduce the dependence on sensors as much as possible during the online testing. A sparse sensing scheme is defined as using only several of the sensors under full sensing. In order to find the optimal sensing scheme, we will study the modeling performance under the conditions of different sensor numbers and those of different sensor placement.

As illustrated in Fig. \ref{fig:sensing_diagram}, during the offline training phase, numerous sensors can be uniformly arranged in the spatial domain (full sensing scheme). However, it is crucial to minimize dependence on sensors during the online testing phase. Hence, a sparse sensing scheme is defined, utilizing only a subset of sensors from the full sensing setup. To determine the optimal sensing scheme, we will investigate the modeling performance under varying sensor quantities and different sensor placements.

\begin{figure}[htbp]
	\centering
	\subfigure[One sensor]{
		\includegraphics[width=0.23\textwidth]{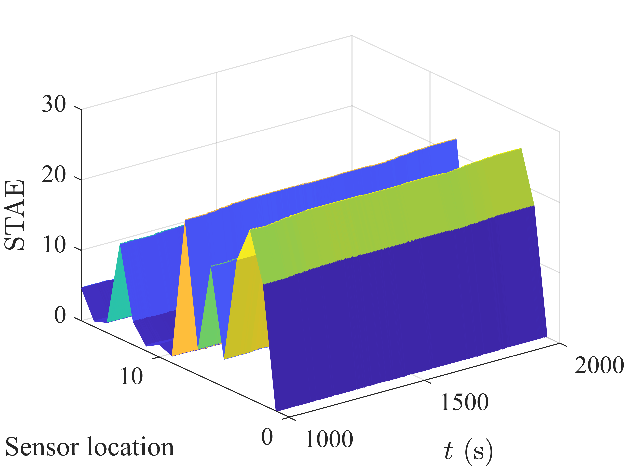}}
	\subfigure[Two sensors]{
		\includegraphics[width=0.23\textwidth]{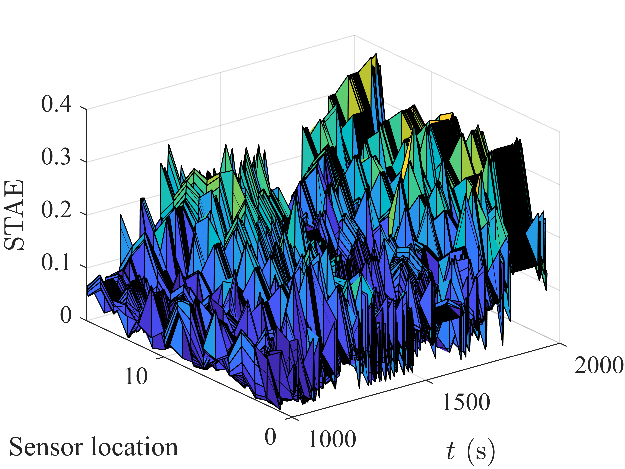}}
	\subfigure[Three sensors]{
		\includegraphics[width=0.23\textwidth]{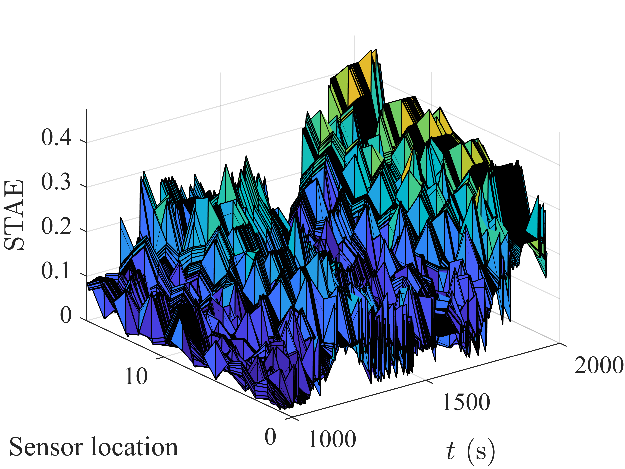}}
	\subfigure[Four sensors]{
		\includegraphics[width=0.23\textwidth]{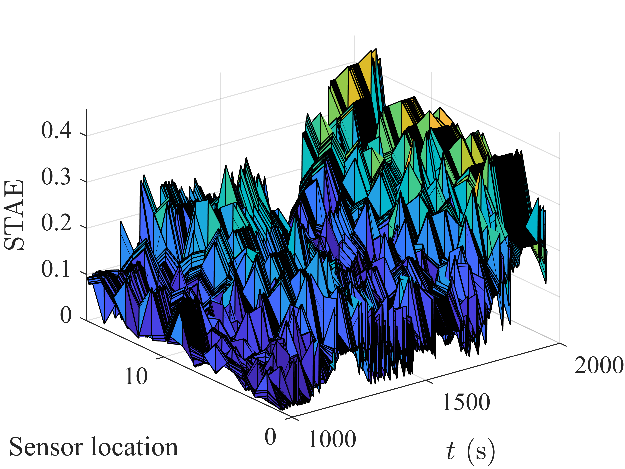}}
	\subfigure[Five sensor]{
		\includegraphics[width=0.23\textwidth]{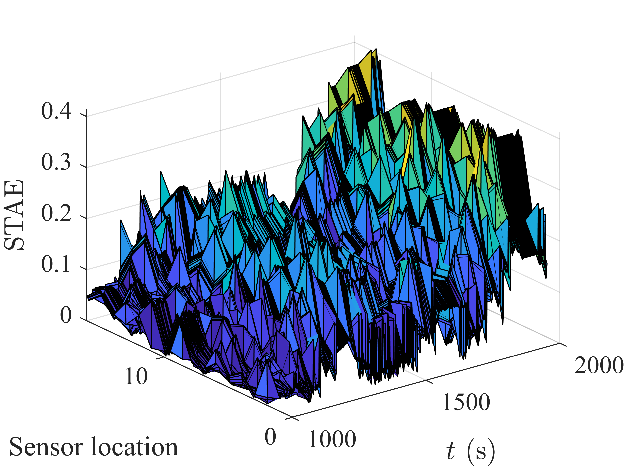}}
	\subfigure[Six sensors]{
		\includegraphics[width=0.23\textwidth]{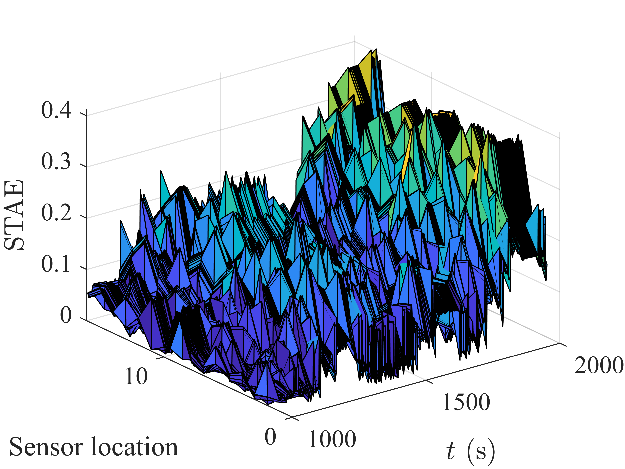}}
	\subfigure[Seven sensors]{
		\includegraphics[width=0.23\textwidth]{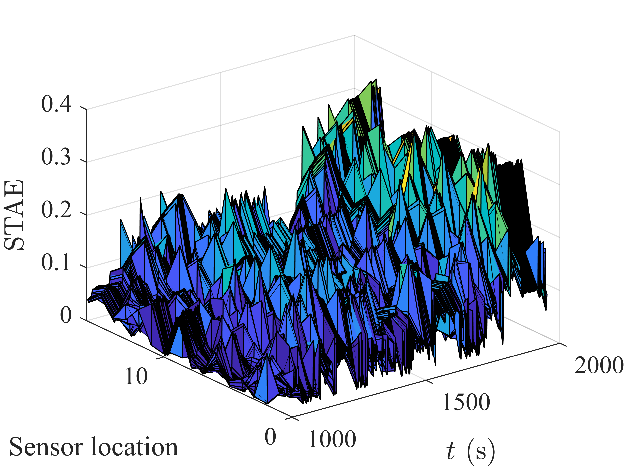}}
	\subfigure[Eight sensors]{
		\includegraphics[width=0.23\textwidth]{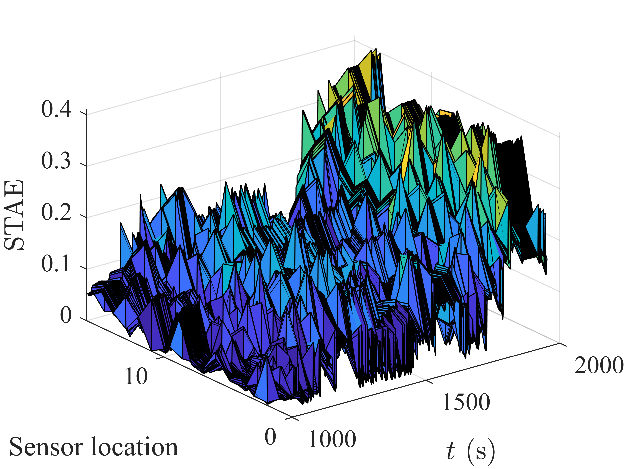}}
	\subfigure[Nine sensor]{
		\includegraphics[width=0.23\textwidth]{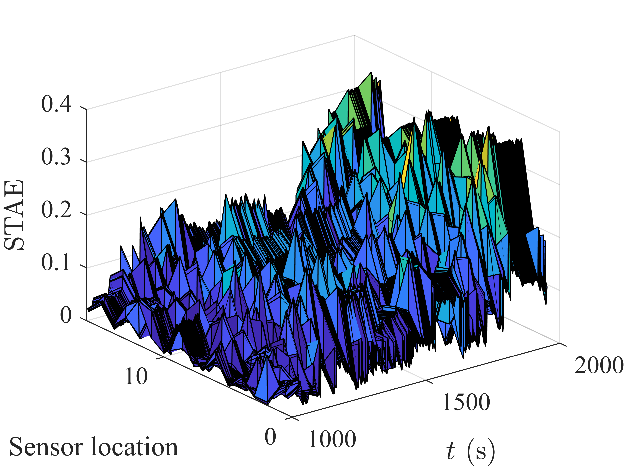}}
	\subfigure[Ten sensors]{
		\includegraphics[width=0.23\textwidth]{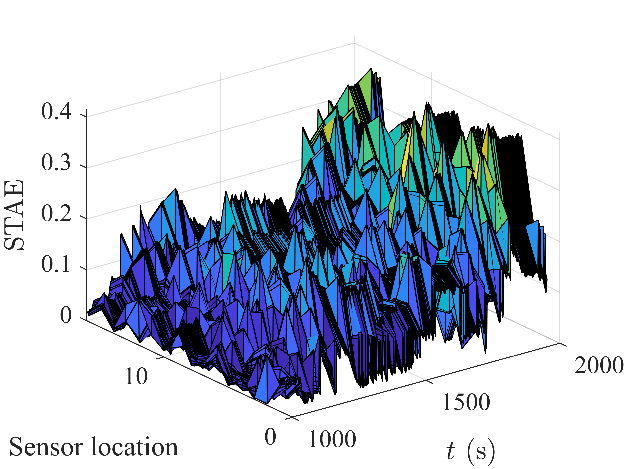}}
	\subfigure[11 sensors]{
		\includegraphics[width=0.23\textwidth]{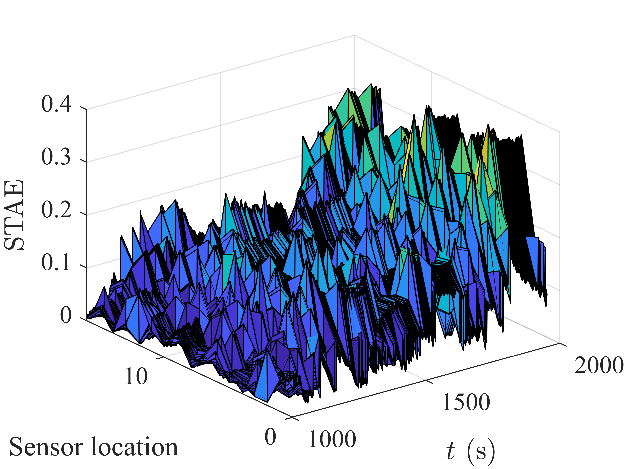}}
	\subfigure[12 sensors]{
		\includegraphics[width=0.23\textwidth]{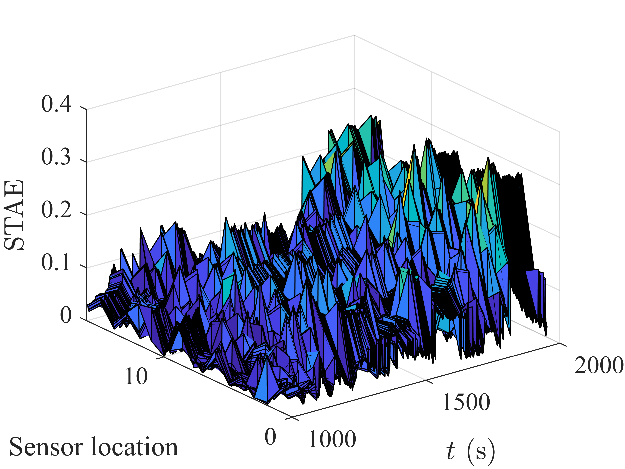}}
	\subfigure[13 sensor]{
		\includegraphics[width=0.23\textwidth]{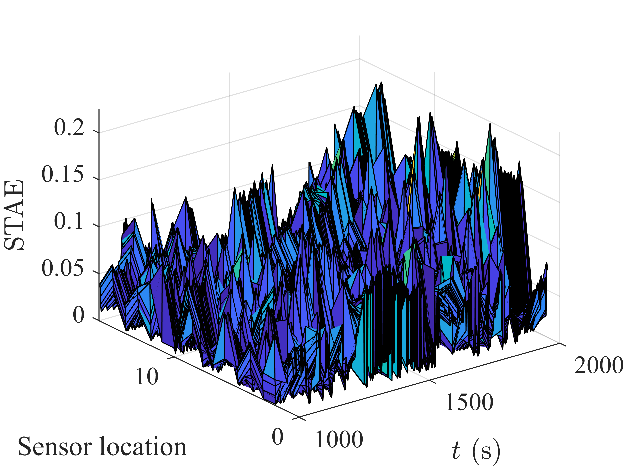}}
	\subfigure[14 sensors]{
		\includegraphics[width=0.23\textwidth]{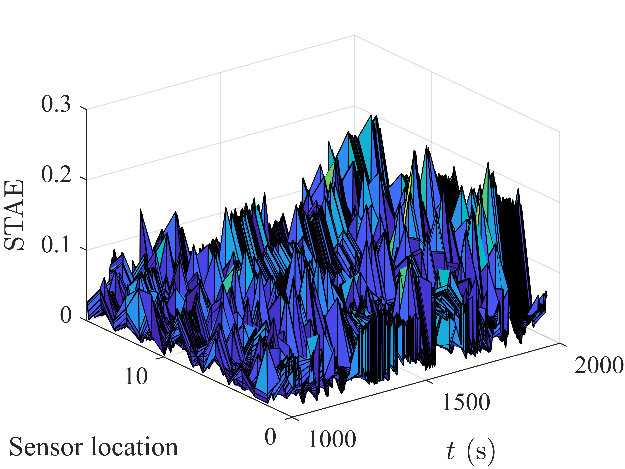}}
	\subfigure[15 sensors]{
		\includegraphics[width=0.23\textwidth]{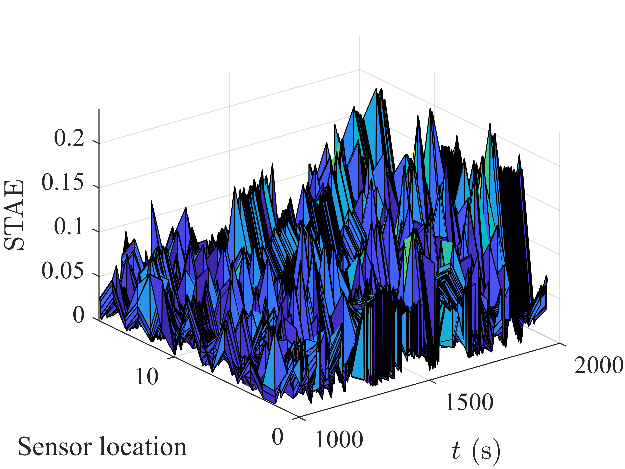}}
	\subfigure[16 sensors]{
		\includegraphics[width=0.23\textwidth]{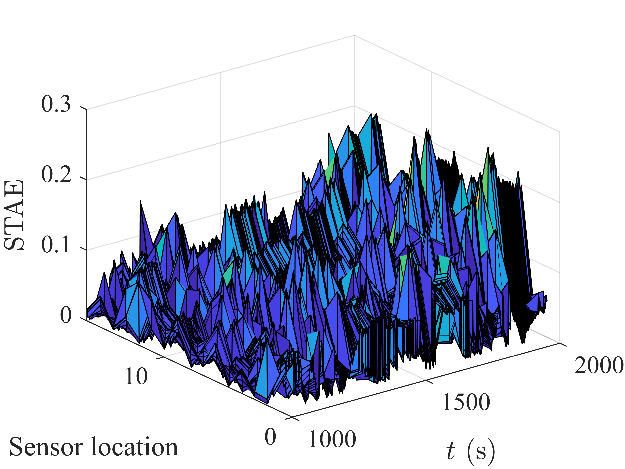}}
	\caption{Spatiotemporal absolute error (STAE) from 1,000 to 2,000 s under different number of sensors. In certain regions, the STAEs reach values as high as 20 under a single sensor configuration, indicating that the thermal dynamics of the battery cannot be adequately captured with only one sensor during the online phase. Interestingly, it is noteworthy that the STAEs between two and sixteen sensors do not exhibit significant variations.}
	\label{fig:RMSEs_different_sensor_number}
\end{figure}

The testing Spatiotemporal absolute errors (STAEs) under different number of sensors are presented in Fig. \ref{fig:RMSEs_different_sensor_number}, with STAE defined as in Equation (\ref{equ:STAE}). The specific sensor locations for corresponding sensor numbers are detailed in Table \ref{tab:different_sesnors_number}. The STAEs exhibit values of up to 20 in certain regions under the one-sensor condition, indicating that the battery thermal process cannot be effectively modeled with just one sensor during the online stage. Additionally, the STAEs between two and 16 sensors exhibit minimal variation, as depicted in Fig. \ref{fig:RMSEs_different_sensor_number}. For a clearer comparison of modeling performance under different sensor quantities, the training and testing RMSEs are presented in Table \ref{tab:different_sesnors_number} and Fig. \ref{fig:RMSE_comparison_different_number_sensors}. The RMSE is computed according to Equation (\ref{equ:RMSE}). The best performance is indicated in bold within Table \ref{tab:different_sesnors_number}. As illustrated in Fig. \ref{fig:RMSE_comparison_different_number_sensors} (a), the RMSE under one sensor significantly surpasses those under multiple sensors, thereby making the differences among the RMSEs under multiple sensors less discernible. Consequently, we plot an enlarged segment of the RMSEs under 2$\sim$16 sensors, as depicted in Fig. \ref{fig:RMSE_comparison_different_number_sensors} (b). The illustration highlights an overall decrease in both training and testing RMSEs as the number of sensors increases, aligning with common intuition. However, to minimize dependency on sensors during online modeling while retaining satisfactory modeling accuracy, we opt to utilize two sensors for the online modeling process.

\begin{figure}[htbp]
	\centering
	\subfigure[Panorama of 1$\sim$16 sensors]{
		\includegraphics[width=0.65\textwidth]{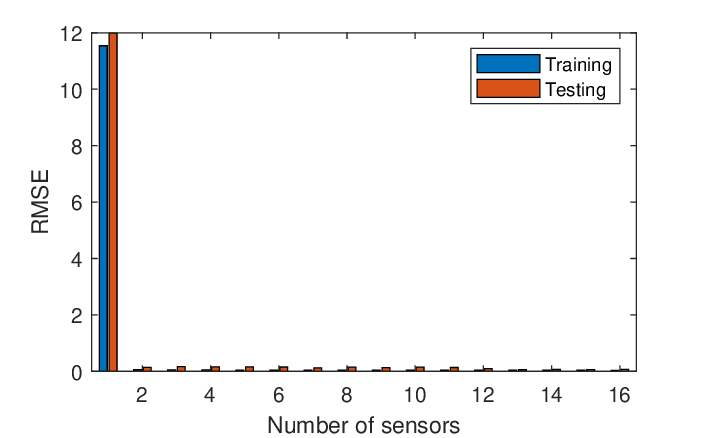}}
	\subfigure[Partial enlargement of 2$\sim$16 sensors]{
		\includegraphics[width=0.65\textwidth]{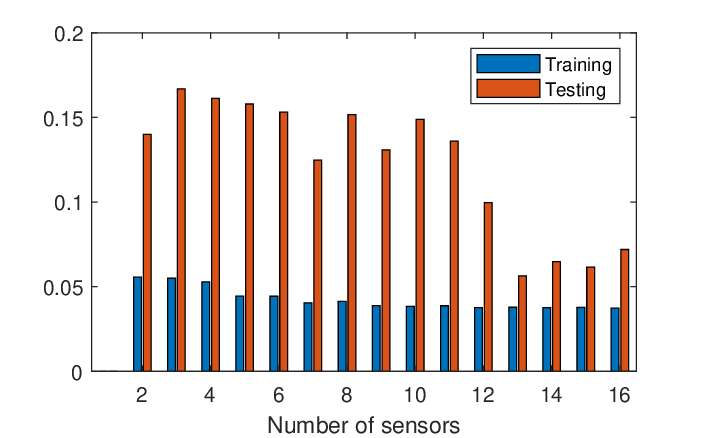}}
	\caption{Root mean square error (RMSE) comparisons under different sensors. (\textbf{a}) Training and testing RMSEs under different number of sensors. (b) Training and testing RMSEs under two to sixteen sensors. The RMSEs associated with a single sensor configuration are substantially greater than those observed with multiple sensors, resulting in less pronounced distinctions among RMSEs under various multi-sensor setups. Consequently, we have created a focused depiction of the RMSEs for configurations ranging from 2 to 16 sensors, as illustrated in Figure (b).}
	\label{fig:RMSE_comparison_different_number_sensors}
\end{figure}

\begin{table}[htbp] 
	\centering
	\caption{Modeling Performance Comparison with Different Number of Sensors.}
	\label{tab:different_sesnors_number}
	\begin{tabular} {cccc}
		\toprule
		Number of sensors  &  \makecell{$\mathbf{s}_\text{tag}$ (Sensor locations)} & Training RMSE &  Testing RMSE \\
		\midrule
		1 & $ [1]^\text{T} $  & 11.5410 & 11.9884 \\
		2 & $ [1,2]^\text{T} $  & 0.0556 & 0.1399 \\
		3 & $ [1,2,3]^\text{T} $  & 0.0550 & 0.1668 \\
		4 & $ [1,2,3,4]^\text{T} $ & 0.0528 & 0.1613 \\
		5 & $ [1,2,3,4,5]^\text{T} $ & 0.0444 & 0.1579 \\
		6 & $ [1,2,3,4,5,6]^\text{T} $  & 0.0443 & 0.1532 \\
		7 & $ [1,2,3,4,5,6,7]^\text{T} $ & 0.0404 & 0.1247 \\
		8 & $ [1,2,3,4,5,6,7,8]^\text{T} $ & 0.0413 & 0.1517 \\
		9 & $ [1,2,3,4,5,6,7,8,9]^\text{T} $ & 0.0389 & 0.1309 \\
		10& $ [1,2,3,4,5,6,7,8,9,10]^\text{T} $ & 0.0384	& 0.1489 \\	
		11& $ [1,2,3,4,5,6,7,8,9,10,11]^\text{T} $ & 0.0388	& 0.1359 \\
		12& $ [1,2,3,4,5,6,7,8,9,10,11,12]^\text{T} $ & 0.0377	& 0.0996 \\	
		13& $ [1,2,3,4,5,6,7,8,9,10,11,12,13]^\text{T} $ & 0.0379	& \textbf{0.0564} \\	
		14& $ [1,2,3,4,5,6,7,8,9,10,11,12,13,14]^\text{T} $ & 0.0377	& 0.0648 \\	
		15& $ [1,2,3,4,5,6,7,8,9,10,11,12,13,14,15]^\text{T} $ & 0.0377	& 0.0615 \\	
		16& $ [1,2,3,4,5,6,7,8,9,10,11,12,13,14,15,16]^\text{T} $ & \textbf{0.0374} & 0.0720 \\					
		\bottomrule
	\end{tabular}
	\begin{tablenotes}
		\item[] The best performance is marked in bold entity.
	\end{tablenotes}
\end{table}

%Since the temperature of the battery is not uniform, sensor locations will affect the performance of the modeling method. In order to find the optimal placement using two sensors, we investigate the STAEs under different sensing schemes as shown in Fig. \ref{fig:RMSEs_different_sensing_schemes}. The temperature near the positive terminal of the battery is high and fluctuates widely. We use sensor \#3 in every sensing scheme due to its proximity to the battery positive terminal. According to \ref{fig:RMSEs_different_sensing_schemes}, it is difficult to find the optimal sensing scheme because the differences among these STAEs are not much obvious. 

Considering the non-uniform temperature distribution within the battery, the positioning of sensors significantly impacts the performance of the modeling method. To determine the optimal placement using two sensors, we analyze the STAEs under various sensing schemes, depicted in Figure \ref{fig:RMSEs_different_sensing_schemes}. Notably, due to the high and fluctuating temperatures near the positive terminal of the battery, sensor \#3 is consistently included in each sensing scheme owing to its proximity to the battery positive terminal. As illustrated in Fig. \ref{fig:RMSEs_different_sensing_schemes}, identifying the optimal sensing scheme proves challenging, as the differences among the various STAEs are not prominently distinct.

%Table \ref{tab:sensing_schemes} compares the modeling performances under different sensing schemes. The best performance is marked in bold entity. The best training performance (lowest RMSE) is achieved under 11th sensing scheme $ \mathbf{s}_\text{tag} = [3,12]^\text{T} $, while the best testing performance is achieved under 10th sensing scheme with $ \mathbf{s}_\text{tag} = [3,11]^\text{T} $. The results in Table \ref{tab:sensing_schemes} are also drawn in Fig. \ref{fig:RMSE_comparison_different_sensing_schemes} to observe the trend of RMSEs under different sensing schemes more visually. It can be seen that although two sensors are used, the impact of different sensing schemes on RMSEs is still relatively obvious. All test RMSEs are larger than the corresponding training RMSEs, except the 7th scheme. During the online procedure, the testing RMSE can better reflect the actual modeling performance. Thus, the testing RMSE becomes the main consideration for sensor location selection. The subsequent experiments are carried out under the 10th sensing scheme due to the lowest testing RMSE.

Table \ref{tab:sensing_schemes} provides a comparison of the modeling performances under different sensing schemes, with the best performance indicated in bold. The most optimal training performance (lowest RMSE) is observed under the 11th sensing scheme with $ \mathbf{s}\text{tag} = [3,12]^\text{T} $, while the best testing performance is achieved under the 10th sensing scheme with $ \mathbf{s}\text{tag} = [3,11]^\text{T} $. The results in Table \ref{tab:sensing_schemes} are also presented in Fig. \ref{fig:RMSE_comparison_different_sensing_schemes} to facilitate a more visual observation of the RMSE trends under different sensing schemes. It is apparent that despite the utilization of two sensors, the influence of various sensing schemes on RMSEs remains noticeable. With the exception of the 7th scheme, all test RMSEs are higher than their corresponding training RMSEs. Therefore, during the online procedure, the testing RMSE serves as the primary criterion for sensor location selection. Consequently, the subsequent experiments are conducted under the 10th sensing scheme, attributed to its lowest testing RMSE.

\begin{table}[htbp] 
	\centering
	\caption{Modeling Performance Comparison with Two Sensors under Different Sensing Schemes.}
	\label{tab:sensing_schemes}
	\begin{tabular} {cccc}
		\toprule
		Scheme  &  \makecell{$\mathbf{s}_\text{tag}$ (Sensor locations)} & Training RMSE &  Testing RMSE \\
		\midrule
		1 & $ [1,3]^\text{T} $  & 0.0562 & 0.1673 \\
		2 & $ [2,3]^\text{T} $  & 0.1379 & 0.2124 \\
		3 & $ [3,4]^\text{T} $  & 0.1903 & 0.2051 \\
		4 & $ [3,5]^\text{T} $ & 0.0463 & 0.1002 \\
		5 & $ [3,6]^\text{T} $ & 0.1099 & 0.1372 \\
		6 & $ [3,7]^\text{T} $  & 0.0476 & 0.1094 \\
		7 & $ [3,8]^\text{T} $ & 0.3511 & 0.1865 \\
		8 & $ [3,9]^\text{T} $ & 0.0475 & 0.1133 \\
		9 & $ [3,10]^\text{T} $ & 0.0458 & 0.0828 \\
		10& $ [3,11]^\text{T} $ & 0.0457	& \textbf{0.0692} \\	
		11& $ [3,12]^\text{T} $ & \textbf{0.0435}	& 0.0863 \\
		12& $ [3,13]^\text{T} $ & 0.0440	& 0.1123 \\	
		13& $ [3,14]^\text{T} $ & 0.0449	& 0.0772 \\	
		14& $ [3,15]^\text{T} $ & 0.0494	& 0.0884 \\	
		15& $ [3,16]^\text{T} $ & 0.0433	& 0.1249 \\					
		\bottomrule
	\end{tabular}
	\begin{tablenotes}
		\item[] The best performance is marked in bold entity.
	\end{tablenotes}
\end{table}

\begin{figure}[htbp]
	\centering
	\subfigure[$\mathbf{s}_\text{tag} = \left\{1,3\right\}^\text{T}$]{
		\includegraphics[width=0.23\textwidth]{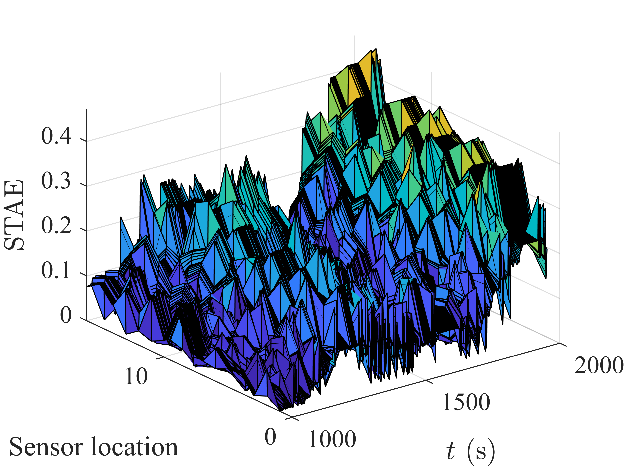}}
	\subfigure[$\mathbf{s}_\text{tag} = \left\{2,3\right\}^\text{T}$]{
		\includegraphics[width=0.23\textwidth]{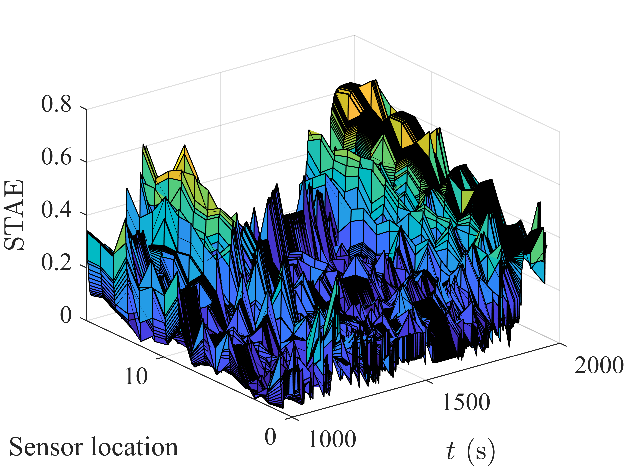}}
	\subfigure[$\mathbf{s}_\text{tag} = \left\{3,4\right\}^\text{T}$]{
		\includegraphics[width=0.23\textwidth]{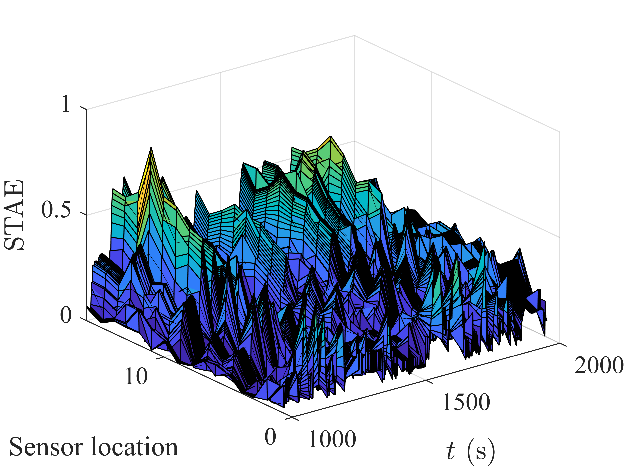}}
	\subfigure[$\mathbf{s}_\text{tag} = \left\{3,5\right\}^\text{T}$]{
		\includegraphics[width=0.23\textwidth]{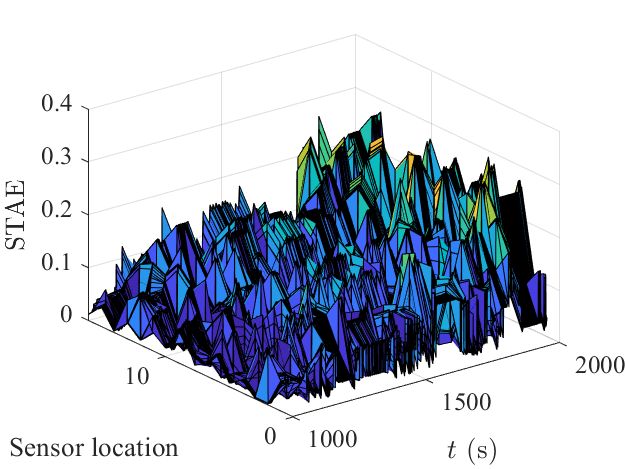}}
	\subfigure[$\mathbf{s}_\text{tag} = \left\{3,6\right\}^\text{T}$]{
		\includegraphics[width=0.23\textwidth]{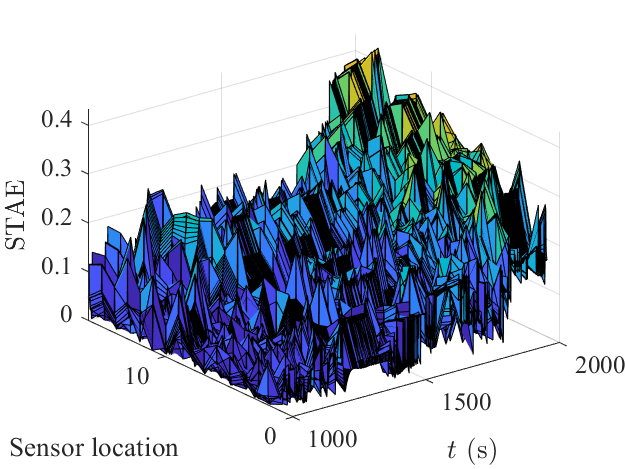}}
	\subfigure[$\mathbf{s}_\text{tag} = \left\{3,7\right\}^\text{T}$]{
		\includegraphics[width=0.23\textwidth]{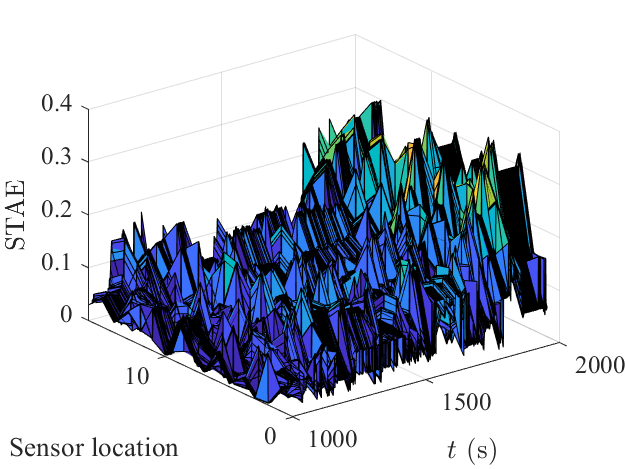}}
	\subfigure[$\mathbf{s}_\text{tag} = \left\{3,8\right\}^\text{T}$]{
		\includegraphics[width=0.23\textwidth]{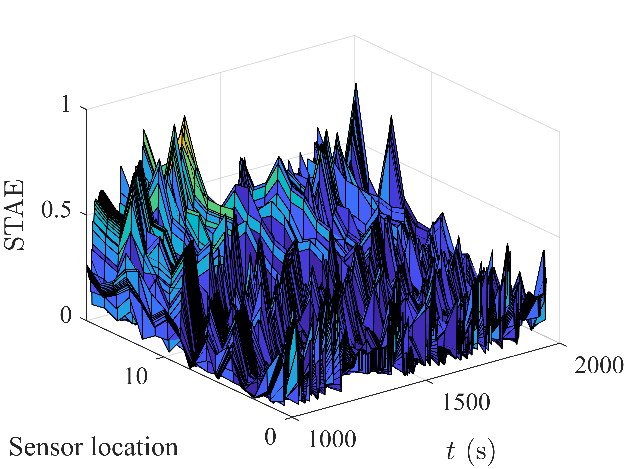}}
	\subfigure[$\mathbf{s}_\text{tag} = \left\{3,9\right\}^\text{T}$]{
		\includegraphics[width=0.23\textwidth]{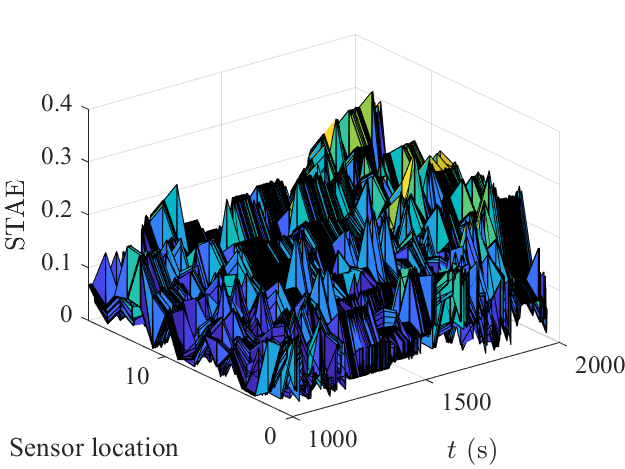}}
	\subfigure[$\mathbf{s}_\text{tag} = \left\{3,10\right\}^\text{T}$]{
		\includegraphics[width=0.23\textwidth]{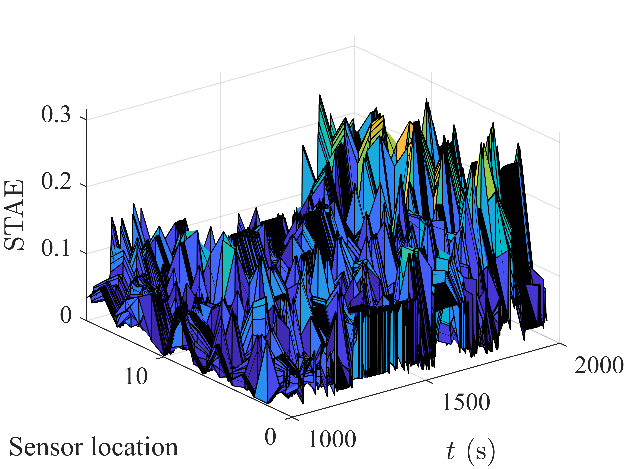}}
	\subfigure[$\mathbf{s}_\text{tag} = \left\{3,11\right\}^\text{T}$]{
		\includegraphics[width=0.23\textwidth]{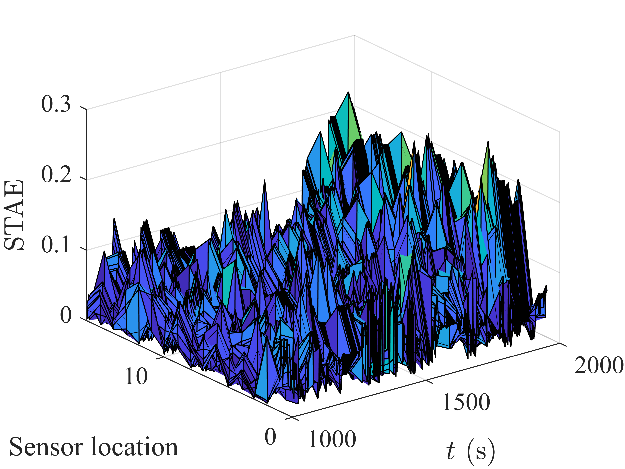}}
	\subfigure[$\mathbf{s}_\text{tag} = \left\{3,12\right\}^\text{T}$]{
		\includegraphics[width=0.23\textwidth]{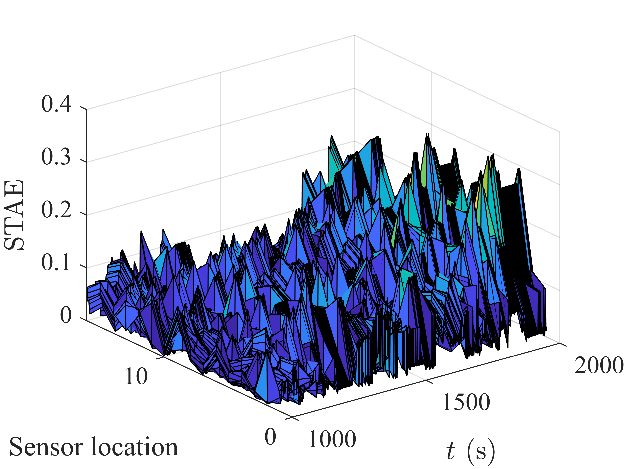}}
	\subfigure[$\mathbf{s}_\text{tag} = \left\{3,13\right\}^\text{T}$]{
		\includegraphics[width=0.23\textwidth]{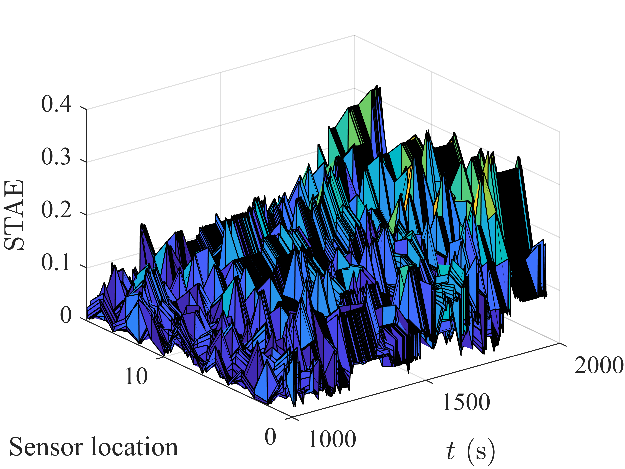}}
	\subfigure[$\mathbf{s}_\text{tag} = \left\{3,14\right\}^\text{T}$]{
		\includegraphics[width=0.23\textwidth]{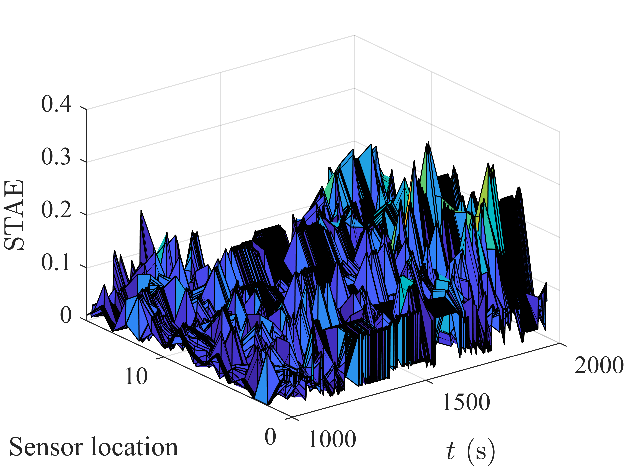}}
	\subfigure[$\mathbf{s}_\text{tag} = \left\{3,15\right\}^\text{T}$]{
		\includegraphics[width=0.23\textwidth]{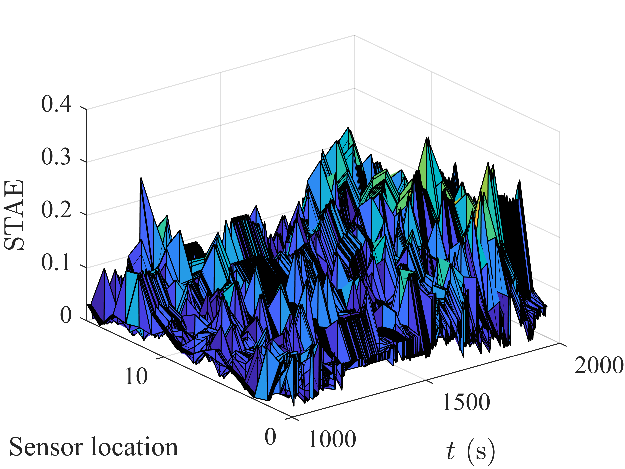}}
	\subfigure[$\mathbf{s}_\text{tag} = \left\{3,16\right\}^\text{T}$]{
		\includegraphics[width=0.23\textwidth]{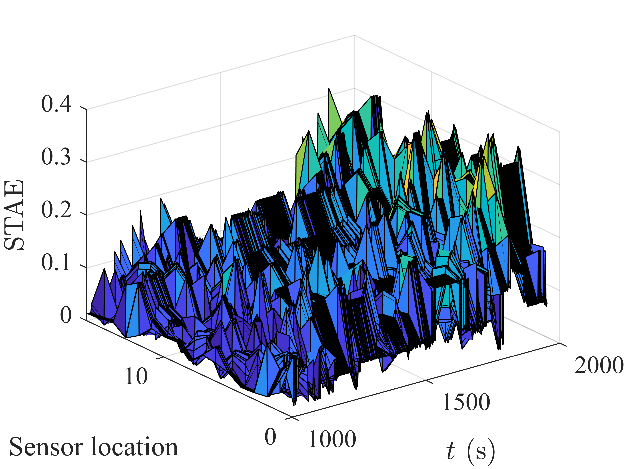}}
	\caption{Spatiotemporal absolute error (STAE) with two sensors under different sensing schemes. In general, there is a discernible upward trend in the STAEs as time progresses. Nonetheless, it remains challenging to definitively determine the most optimal sensing solution based solely on the STAE distribution diagram provided above.}
	\label{fig:RMSEs_different_sensing_schemes}
\end{figure}

\begin{figure}[htbp] 
	\centering
	\includegraphics[width=0.65\textwidth]{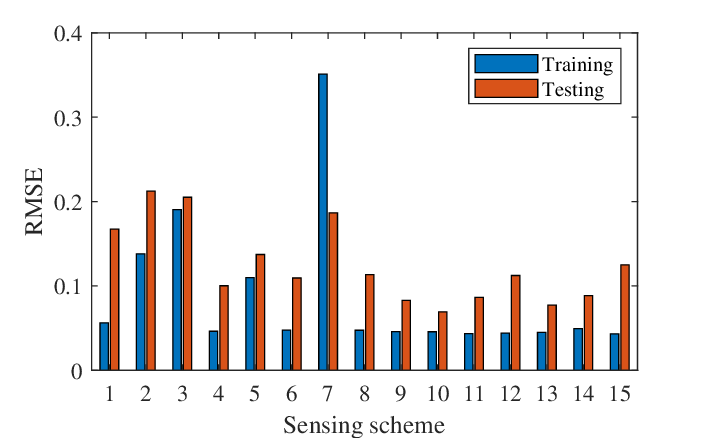}
	\caption{RMSE results under different sensing schemes. Evidently, even when employing only two sensors, the influence of varying sensing schemes on RMSEs remains noticeably discernible. It is noteworthy that, except for the 7th scheme, all test RMSEs exceed their corresponding training RMSEs. This observation underscores the fact that, during the online phase, the testing RMSE more accurately mirrors the actual modeling performance. Consequently, the testing RMSE assumes a primary role in the decision-making process for sensor placement. In light of this, the subsequent experiments are conducted under the 10th sensing scheme, which exhibits the lowest testing RMSE.}
	\label{fig:RMSE_comparison_different_sensing_schemes}	
\end{figure}

\begin{figure}[htbp]
	\centering
	\subfigure[First SBF $ \psi_1(x,y) $]{
		\includegraphics[width=0.4\textwidth]{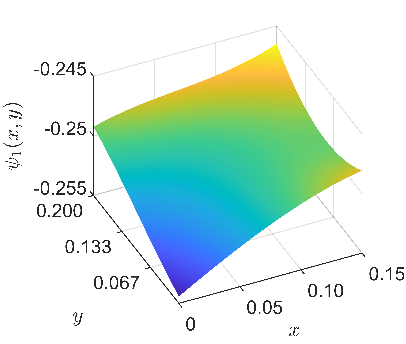}}
	\subfigure[Second SBF $ \psi_2(x,y) $]{
		\includegraphics[width=0.4\textwidth]{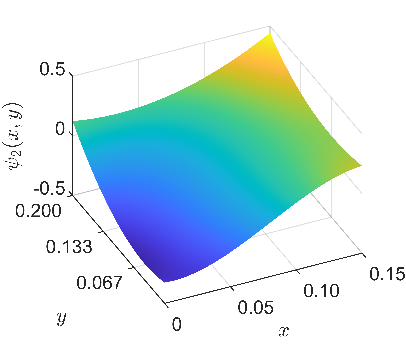}}
	\caption{Constructed spatially continuous SBFs. During the online process, despite the utilization of only two sensors, the capability to discern all spatial features remains intact. This assurance lays the groundwork for subsequent spatiotemporal predictions encompassing the entirety of the spatial domain.}
	\label{fig:SBFs}
\end{figure}

%According to the proposed spatial construction scheme, the first two spatially continuous SBFs are derived, as shown in Fig. \ref{fig:SBFs}. It should be noted that only two sensors are used for data acquisition during the online stage, as shown in Fig. \ref{fig:sensing_diagram}(b). Benefited from spatially continuous SBFs, the full space temperature distributions can be predicted under sparse sensing, as shown in Fig. \ref{fig:prediction_1500s} and Fig. \ref{fig:prediction_2000s}. The temperature of the area close to the positive pole of the battery cell is higher than that of other areas, which is also in line with the actual situation. 

In accordance with the proposed space construction scheme, the initial two spatially continuous SBFs are derived, illustrated in Fig. \ref{fig:SBFs}. Notably, during the online stage, only two sensors are utilized for data acquisition, as demonstrated in Fig. \ref{fig:sensing_diagram}(b). Leveraging the spatially continuous SBFs enables the prediction of complete spatial temperature distributions even under conditions of sparse sensing, as depicted in Figures \ref{fig:prediction_1500s} and \ref{fig:prediction_2000s}. 
Indeed, the temperature near the positive pole of the battery cell exceeds that of other regions, aligning with the actual observed patterns.

%The STAE distributions are also added aside the corresponding temperature distributions. Overall, the STAEs at 2,000 s are higher those at 1,500 s. In order to quantitatively describe the change of modeling error over time, the SNAE curve from 0 $\sim$ 2,000 s is shown in Fig. \ref{fig:SNAE}. The SNAE is defined in Equ. (\ref{equ:SNAE}). During the training phase, the change of SNAE is relatively smooth. During the testing phase, SNAE fluctuates and tends to increase over time.  

The STAE distributions are provided alongside their corresponding temperature distributions. Overall, the STAEs at 2,000 s exhibit higher values compared to those at 1,500 s. To quantitatively depict the variation in modeling error over time, the spatial normalized absolute error (SNAE) curve from 0 $\sim$ 2,000 s is presented in Fig. \ref{fig:SNAE}. The SNAE is defined as in Equation (\ref{equ:SNAE}). Notably, during the training phase, the SNAE demonstrates a relatively smooth transition. Conversely, during the testing phase, the SNAE displays fluctuations and an inclination towards an increase over time.

\begin{figure}[htbp] 
	\centering
	\includegraphics[width=0.8\textwidth]{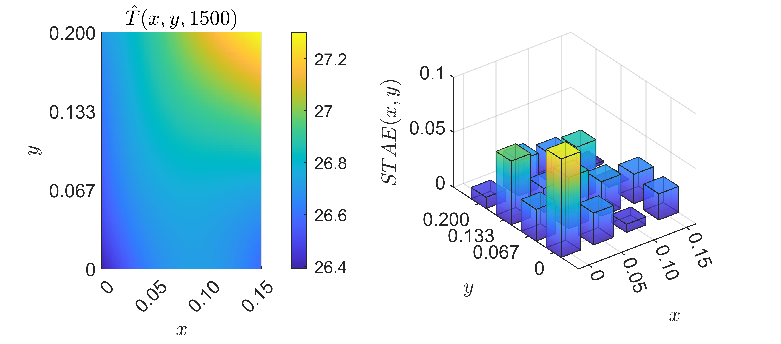}
	\caption{Full space prediction with corresponding STAE at 1,500 s.}
	\label{fig:prediction_1500s}	
\end{figure}

\begin{figure}[htbp] 
	\centering
	\includegraphics[width=0.8\textwidth]{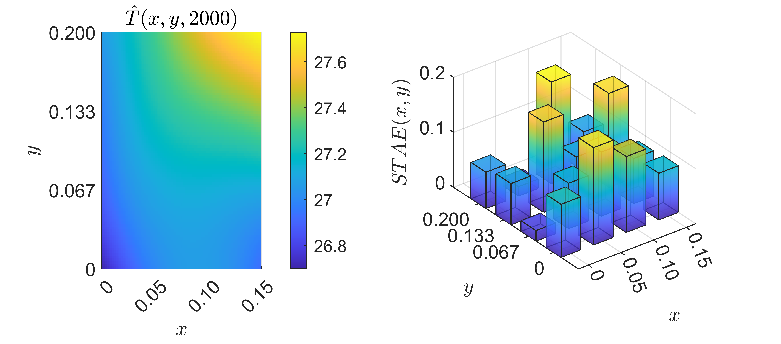}
	\caption{Full space prediction with corresponding STAE at 2,000 s.}
	\label{fig:prediction_2000s}	
\end{figure}

\begin{figure}[htbp] 
	\centering
	\includegraphics[width=0.8\textwidth]{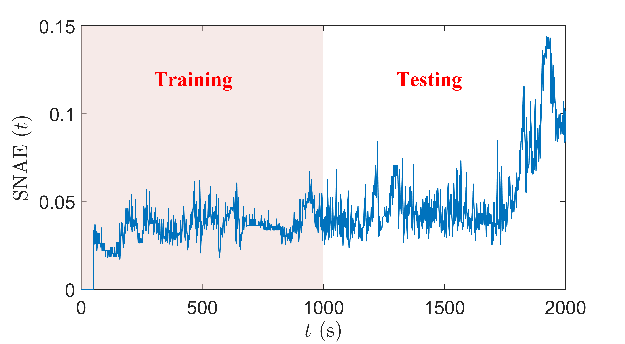}
	\caption{Spatial normalized absolute error (SNAE) from 0 $\sim$ 2,000 s. Prior to the 1800s, the SNAE exhibited fluctuations centered around 0.05. However, starting from the 1800s, there is a discernible and steep upward trajectory in SNAE.}
	\label{fig:SNAE}	
\end{figure}

%As shown in Table \ref{tab:performance_comparison}, the conventional KL \cite{qi2011time}, the sliding-window KL \cite{wang2018sliding}, and sparse KL \cite{chen2020spatiotemporal} are used to compared with the proposed method. The conventional KL and the SW-KL have better modeling performance in training and testing RMSEs since more sensors are used. However, they cannot work under sparse sensing and cannot achieve full space prediction. Although the sparse KL can work in the case of sparse sensing, its modeling accuracy is not satisfactory and it cannot achieve the full space prediction. Benefited from the proposed space construction, the proposed method can achieve full space prediction under sparse sensing. Since the LSTM algorithm can update the temporal model with the latest data and the constructed SBFs are insensitive to noise, the proposed method has smaller training and testing RMSEs compared with the sparse KL under the same sensing condition. Overall, the proposed method outperforms other peer methods in terms of RMSE under the same sensing conditions and can achieve full space prediction under sparse sensing.  

In Table \ref{tab:performance_comparison}, the conventional Karhunen-Loeve (KL) method \cite{qi2011time}, sliding-window KL method \cite{wang2018sliding}, and sparse KL method \cite{chen2020spatiotemporal} are compared with the proposed method. The conventional KL and the SW-KL demonstrate superior modeling performance in terms of training and testing RMSEs, owing to their utilization of a higher number of sensors. However, they are not viable under sparse sensing conditions and are incapable of achieving full-space prediction. Although the sparse KL method can work with sparse sensing, its modeling accuracy is unsatisfactory, and it fails to achieve complete space prediction. Leveraging the proposed space construction methodology, the proposed method successfully achieves full space prediction even under sparse sensing. As the LSTM algorithm can continuously update the temporal model with the latest data and the constructed SBFs are relatively immune to noise, the proposed method exhibits smaller training and testing RMSEs in comparison to the sparse KL method under the same sensing conditions. Overall, the proposed method surpasses other comparative methodologies in terms of RMSE under identical sensing conditions, while successfully achieving full space prediction under sparse sensing conditions.

\begin{table}[htbp]
	\centering
		\caption{Performance Comparison of Different Methods.}
		\label{tab:performance_comparison}
		\begin{tabular} {ccccc} 
			\toprule 
			& \makecell{Training\\RMSE} & \makecell{Testing\\RMSE} &\makecell{Work under\\sparse sensing} & \makecell{Full space\\prediction}\\
			\midrule
			Conventional KL   & 0.0238  & 0.0475  & \XSolid\  (16 sensors) & \XSolid \\
			SW-KL   & 0.0217  & 0.0391  & \XSolid\ (16 sensors) & \XSolid \\	
			Sparse KL   & 0.0623 & 0.1494  & \Checkmark\ (2 sensors) & \XSolid \\	
			\textbf{Proposed method} & \textbf{0.0457} & \textbf{0.0692} & \Checkmark\ \textbf{(2 sensors)} & \Checkmark\\
			\bottomrule
		\end{tabular}
		\begin{tablenotes}
			\item[] The performance of the proposed method is marked in bold entity.
		\end{tablenotes}
\end{table}

\section{Discussion}\label{sec3}
%In summary, the research presented in this study is dedicated to addressing the challenge of achieving full-space modeling of DPSs in sensor-limited scenarios. Our proposed two-stage space construction methodology is validated through a series of experiments involving Li-ion batteries. We also investigate the influence of varying sensor quantities on the modeling performance of our proposed framework. Notably, our findings indicate that the performance of our proposed method demonstrates limited sensitivity to the number of sensors when the count exceeds two. This observation paves the way for achieving full-space modeling under sparse sensing. The experimental results underscore the effectiveness of our proposed approach, as it achieves full-space prediction of battery thermal processes, attaining a training RMSE of 0.0457 and a testing RMSE of 0.0692, all while utilizing only two online sensors.

In summary, the research presented in this study is dedicated to addressing the challenge of achieving full-space modeling of DPSs in sensor-limited scenarios. Our proposed two-stage space construction methodology is rigorously validated through a series of experiments centered on the thermal processes of Li-ion batteries. Furthermore, we thoroughly investigate the impact of varying sensor quantities on the modeling performance of our proposed framework. Notably, our results suggest that the performance of our proposed method exhibits limited sensitivity to the number of sensors once the count surpasses two, thus laying the foundation for achieving full-space modeling under conditions of sparse sensing. The experimental outcomes serve to underscore the efficacy of our proposed approach, as it successfully achieves full-space prediction of battery thermal processes, attaining a training RMSE of 0.0457 and a testing RMSE of 0.0692, all while employing only two online sensors.

Our analysis reveals that even with only two sensors, the specific locations chosen for sensor placement exert a discernible influence on modeling performance. This observation underscores the need for further efforts on optimizing sensor positioning strategies. Moreover, while our work has focused on the application of space construction method for sparse-sensing based online modeling of BESSs, such framework is promising for other full-space modeling and prediction tasks under sparse sensing, particularly those for spatiotemporal dynamical systems such as chip curing processes, chemical reaction processes, and robotic arm systems.

\section{Methods}\label{sec4}
\subsection{Discrete space completion under sparse sensing}
%Take the two-dimensional space as an example. As shown in Fig. \ref{fig:conceptio n}, assume $ N_\text{f} $ and $ N_\text{s} $ $ (N_\text{f}\geq N_\text{s}) $ sensors are used for data collection under full sensing and sparse sensing, respectively. Suppose $ l_1 $ number of data snapshots are collected during the offline stage. The mapping between sparse and full sensing can be established as follows:

Let's consider a two-dimensional space, as illustrated in Fig. \ref{fig:framework}. We assume the usage of $ N_\text{f} $ and $ N_\text{s} $ sensors, where $ N_\text{f} \geq N_\text{s} $, for data collection under conditions of full sensing and sparse sensing, respectively. During the offline stage, let's suppose $ l_1 $ data snapshots are collected. The establishment of a mapping between sparse and full sensing is outlined as follows:
\begin{equation}\label{equ:mapping}
	\mathbf{T}_\text{s} = \mathbf{M} \mathbf{T}_\text{f}
\end{equation}
%where $ \mathbf{T}_\text{s} \in \mathcal{R}^{N_\text{s}\times l_1}$ and $ \mathbf{T}_\text{f} = [\mathbf{s}(1),\mathbf{s}(2),\ldots,\mathbf{s}(l_1)]^\text{T} \in \mathcal{R}^{N_\text{f}\times l_1}$ denote the data matrix corresponding to sparse and full sensing, respectively. The mapping matrix $ \mathbf{M} \in \mathcal{R}^{N_\text{s}\times N_\text{f}}$ is defined as:
where $ \mathbf{T}_\text{s} \in \mathcal{R}^{N_\text{s}\times l_1}$ and $ \mathbf{T}_\text{f} = [\mathbf{s}(1),\mathbf{s}(2),\ldots,\mathbf{s}(l_1)]^\text{T} \in \mathcal{R}^{N_\text{f}\times l_1}$ represent the data matrix corresponding to sparse and full sensing, respectively. The mapping matrix $ \mathbf{M} \in \mathcal{R}^{N_\text{s}\times N_\text{f}}$ is defined as:
\begin{equation}\label{}
	\mathbf{M}=(m_{ij})\left\{\begin{array}{l}
		1,\ \text{if}\ (i,j)\in \{(i,s_i)\}_{i=1}^{N_\text{s}} \\
		0,\ \text{otherwise}
	\end{array}\right.
\end{equation}
%where $ s_i $ denotes the $ i $th tag number under sparse sensing. For example, the tag vector corresponding to Fig. 1(a) can be expressed as $ \mathbf{s}_\text{tag}  = [s_1,s_2,s_3,s_4]^\text{T} = [2,8,11,13]^\text{T}$.
where $ s_i $ denotes the $ i $th tag number under sparse sensing. For instance, the tag vector corresponding to Fig. 1(a) can be expressed as $ \mathbf{s}_\text{tag}  = [s_1,s_2,s_3,s_4]^\text{T} = [2,8,11,13]^\text{T}$.

In the effort to extract the temporal and spatial dynamics, the measurement matrices can be decomposed through space-time separation \cite{zhou2023data} as follows:  
\begin{equation}\label{equ:sparse_KL}
	\mathbf{T}_\text{s} = \mathbf{\Phi}_\text{s} \mathbf{a}_\text{s}
\end{equation}
\begin{equation}\label{equ:full_KL}
	\mathbf{T}_\text{f} = \mathbf{\Phi}_\text{f} \mathbf{a}_\text{f}
\end{equation}
where $ \mathbf{\Phi}_\text{s} = [\boldsymbol{\varphi}_1,\boldsymbol{\varphi}_2,\ldots,\boldsymbol{\varphi}_{n_\text{s}}] \in \mathcal{R}^{N_\text{s}\times n_\text{s}}$ is the spatial basis function (SBF) matrix under sparse sensing with $ \boldsymbol{\varphi}_i \in \mathcal{R}^{N_\text{s}} $ denoting the $ i $th sparse SBF; $ \mathbf{\Phi}_\text{f} = [\boldsymbol{\zeta}_1,\boldsymbol{\zeta}_2,\ldots,\boldsymbol{\zeta}_{n_\text{f}}] \in \mathcal{R}^{N_\text{f}\times n_\text{f}}$ is the SBF matrix under full sensing with $ \boldsymbol{\zeta}_i \in \mathcal{R}^{N_\text{f}} $ denoting the $ i $th full SBF; $ \mathbf{\Phi}_\text{s} $ and $ \mathbf{\Phi}_\text{s} $ are both orthogonal matrices, that is, the SBFs are orthogonal to each other in a SBF matrix; $ \mathbf{a}_\text{s} \in \mathcal{R}^{n_\text{s}\times l_1}$ and $ \mathbf{a}_\text{f} \in \mathcal{R}^{n_\text{f}\times l_1}$ are temporal coefficient matrices under sparse and full sensing, respectively; $ n_\text{s} $ and $ n_\text{f} $ represent the model orders under sparse sensing and full sensing, respectively; The model order can be derived according to the first $ n_\text{s} $ (or $ n_\text{f} $) number of SBFs occupying 99\% system energy \cite{li2010modeling} under sparse sensing (or full sensing).

Substituting (\ref{equ:sparse_KL})(\ref{equ:full_KL}) into (\ref{equ:mapping}), the iterative computation of the full temporal coefficient vector can be expressed as:
\begin{equation}\label{equ:derivation_full_temporal_coefficient}
	\mathbf{a}_\text{f} = (\mathbf{\Phi}_\text{s}^\text{T} \mathbf{M} \mathbf{\Phi}_\text{f})^\dagger \mathbf{a}_\text{s} 
\end{equation}
where the symbol $ \dagger $ denotes the pseudo inverse of a matrix. During the online stage, the full data matrix in (\ref{equ:full_KL}) can be expressed as 
\begin{equation}\label{equ:spaital_completion}
	\widehat{\mathbf{T}}_\text{f} = \mathbf{\Phi}_\text{f} \mathbf{a}_\text{f} = \mathbf{\Phi}_\text{f}  (\mathbf{\Phi}_\text{s}^\text{T} \mathbf{M} \mathbf{\Phi}_\text{f})^\dagger \mathbf{a}_\text{s}
\end{equation}
Note that $ \mathbf{\Phi}_\text{s} $ and $ \mathbf{\Phi}_\text{f} $ have been derived by the space-time separation. Therefore, we can use the sparse temporal coefficient matrix $ \mathbf{a}_\text{s} $ to recover the full spatial measurement by (\ref{equ:spaital_completion}).

\subsection{Continuous space construction}
%In order to preserve the interactions between different spatial dimensions, each column of the full SBF matrix $ \mathbf{\Phi}_\text{f} $, i.e., $ \{\boldsymbol{\zeta}\}_{i=1}^{n_\text{f}} $ should be reshaped to a $ (N_1 \times N_2)$ SBF matrix denoted as $\{\boldsymbol{\phi}_i \}_{i=1}^{n_\text{f}} $. Let the model order $ n=n_\text{f} $. The continuous SBFs $ \{\psi_i(x,y)\}_{i=1}^n $ can be constructed under the following optimization:
In order to preserve the interactions between different spatial dimensions, each column of the full SBF matrix $ \mathbf{\Phi}_\text{f} $, i.e., $ \{\boldsymbol{\zeta}\}_{i=1}^{n_\text{f}} $ should be reshaped to a $ (N_1 \times N_2)$ SBF matrix, represented as $\{\boldsymbol{\phi}_i \}_{i=1}^{n_\text{f}} $. Let the model order be denoted as $ n=n_\text{f} $. The continuous SBFs $ \{\psi_i(x,y)\}_{i=1}^n $ can be constructed under the following optimization:
\begin{equation}\label{equ:space_construction_optimization}
	\begin{aligned}
		&\min_{\psi_i(x,y)} \sum_{j_1=1}^{N_1}\sum_{j_2=1}^{N_2} \left[\boldsymbol{\phi}_i(j_1,j_2)-\psi_i(x_{j_1},y_{j_2})\right]^2 \\
		\text{s.t.}\ &\psi_i(x,y)\in C^r(\Omega);\ i=1,2,\ldots,n;\ r=0,1,2\ldots   
	\end{aligned}
\end{equation}
where $ \psi_i(x,y) $ denotes the $ i $th spatially continuous SBF to be designed; $ x $ and $ y $ denote the spatial variables corresponding to the first and second spatial dimensions, respectively; $ N_1 $ and $ N_2 $ signify the number of sensors along the $ x $ and $ y $ directions, respectively; $ \boldsymbol{\phi}_i \in \mathcal{R}^{N_1\times N_2}$ is the $ i $th discrete full SBF matrix; $ C^r(\Omega) $ denotes the high-dimensional function set with continuous $ r $-order partial derivatives along the $ x $ and $ y $ directions, where $\Omega$ represents the entire space domain; and $ n $ is the system model order, selected as the model order under full sensing in practical implementation, i.e., $ n=n_\text{f} $.

%In order to have a proper solution of the optimization problem (\ref{equ:space_construction_optimization}), the high-dimensional continuous SBF $ \psi_i $ should satisfy the following design principles:
In order to ensure an appropriate solution for the optimization problem (\ref{equ:space_construction_optimization}), the high-dimensional continuous SBF $ \psi_i $ should adhere to the following design principles:
\begin{itemize}
	\item[(a)] $ \psi_i $ is required to possess continuous first- and second-order partial derivatives along the $ x $ and $ y $ directions, denoted as $ \psi_i \in C^2(\Omega)$.
	\item[(b)] $ \psi_i $ must be a function of spatial coordinates $ x $ and $ y $ and not a parametric function.
	\item[(c)]  $ \psi_i $ should demonstrate insensitivity to outliers, meaning that deviations in a data point will only affect a portion of the SBF rather than the entire function.
\end{itemize}

\begin{definition}
	For a given function $ f(x,y) \in C(\Omega)$, if there exists a function $ p^{*}(x,y) \in H_n(\Omega) $ such that
	\begin{equation}\label{equ:definition1}
		\left(f-p^{*},f-p^{*}\right) = \min_{p_i \in H_n(\Omega)} \left(f-p_i,f-p_i\right)
	\end{equation}
	then $ p^{*}(x,y) $ is referred to as the least squares approximation element in the subspace $ H_n(\Omega) $, where $ H_n = \text{Span}\{p_1,p_2,\ldots,p_n\} $. Here, $ \left(\cdot,\cdot\right) $ denotes the inner product operator defined as $ \left(f(x,y),g(x,y)\right) \triangleq \int\int\int_\Omega f(x,y)\cdot g(x,y) dxdy$; $ p_i \in C(\Omega) $, $\left(p_i,p_j\right)|_{i \neq j} = 0 $; $ \Omega $ represents the entire space domain. 
\end{definition}

\begin{theorem}\label{thm1}
	Suppose there exists an infinite number of sensors uniformly distributed in the space domain $\Omega$. Consequently, $ \boldsymbol{\phi}_i $ in (\ref{equ:space_construction_optimization}) can be considered as a high-dimensional continuous function in $\Omega$, and the optimization (\ref{equ:space_construction_optimization}) is tantamount to the least squares approximation in (\ref{equ:definition1}). The sufficient and necessary condition for $ \psi_i^{*}(x,y) \in H(\Omega)$ to be the optimal solution of the optimization (\ref{equ:space_construction_optimization}) is that
	\begin{equation}\label{key}
		\left(\boldsymbol{\phi}_i(x,y)-\psi_i^{*}(x,y),\varphi_j(x,y)\right)=0
	\end{equation}
	or for any $ \psi_i(x,y) \in H(\Omega) $, 
	\begin{equation}\label{equ:condition_in_theorem1}
		\left(\boldsymbol{\phi}_i(x,y)-\psi_i^{*}(x,y),\psi_i(x,y)\right) = 0
	\end{equation}
	in which $ H(\Omega)\triangleq\text{Span}\{\varphi_1,\varphi_2,\ldots\}$; $ \varphi_j(x,y) \in C^r(\Omega)$; $\left(\varphi_j,\varphi_k\right)|_{j \neq k} = 0 $; $\left(\varphi_j,\varphi_k\right)|_{j = k} = 1 $; $ i=1,2,\ldots,n $; $ j=1,2,\ldots $; $ r=0,1,2,\ldots $ 
\end{theorem}

\begin{proof}
	1) Proof of necessity: Suppose there exists a high-dimensional continuous function $\varphi_k(x,y,z) \in H(\Omega)$ such that
	\begin{equation}\label{equ:assumption1}
		\left(\boldsymbol{\phi}_i(x,y)-\psi_i^{*}(x,y),\varphi_k(x,y)\right) = \sigma \neq 0
	\end{equation}
	in which $ \sigma $ is a real number.
	Let
	\begin{equation}\label{key}
		q(x,y,z) = \psi_i^{*}(x,y,z) + \frac{\sigma}{(\varphi_k,\varphi_k)}\varphi_k(x,y,z).
	\end{equation}
	Obviously, $ q(x,y,z) \in H(\Omega) $ since $ \psi_i^{*},\varphi_k \in H(\Omega)$. Then, we have
	\begin{equation}\label{key}
		\begin{aligned}
			&(\boldsymbol{\phi}_i-q,\boldsymbol{\phi}_i-q) \\
			=& (\boldsymbol{\phi}_i-\psi_i^{*},\boldsymbol{\phi}_i-\psi_i^{*}) - \frac{2\sigma}{(\varphi_k,\varphi_k)}(\boldsymbol{\phi}_i-\psi_i^{*},\varphi_k)
			+ \frac{\sigma^2}{(\varphi_k,\varphi_k)^2}(\varphi_k,\varphi_k)
		\end{aligned}
	\end{equation}
	According to (\ref{equ:assumption1}), we have
	\begin{equation}\label{key}
		\begin{aligned}
			(\boldsymbol{\phi}_i-q,\boldsymbol{\phi}_i-q) =& (\boldsymbol{\phi}_i-\psi_i^{*},\boldsymbol{\phi}_i-\psi_i^{*}) - \frac{\sigma^2}{(\varphi_k,\varphi_k)^2} \\
			<& (\boldsymbol{\phi}_i-\psi_i^{*},\boldsymbol{\phi}_i-\psi_i^{*})
		\end{aligned}
	\end{equation}
	Consequently, $ \psi_i^{*}(x,y) $ does not represent the least squares approximation element of $ \boldsymbol{\phi}_i $. This contradiction establishes the necessity.
	
	2) Proof of sufficiency: Suppose the condition (\ref{equ:condition_in_theorem1}) holds. Then, for any $ \psi_i \in H(\Omega) $, we have
	\begin{equation}\label{key}
		\begin{aligned}
			&(\boldsymbol{\phi}_i-\psi_i,\boldsymbol{\phi}_i-\psi_i)\\
			=&(\boldsymbol{\phi}_i-\phi_i^{*}+\phi_i^{*}-\psi_i^,\boldsymbol{\phi}_i-\phi_i^{*}+\phi_i^{*}-\psi_i)\\
			=&(\boldsymbol{\phi}_i-\psi_i^{*},\boldsymbol{\phi}_i-\psi_i^{*})+2(\boldsymbol{\phi}_i-\psi_i^{*},\psi_i^{*}-\psi_i)\\
			&+(\psi_i^{*}-\psi_i,\psi_i^{*}-\psi_i)
		\end{aligned}
	\end{equation}
	According to (\ref{equ:condition_in_theorem1}), since $ \psi_i^{*}\in H(\Omega)$, we have
	\begin{equation}\label{key}
		\left(\boldsymbol{\phi}_i-\psi_i^{*},\psi_i^{*}\right) = 0
	\end{equation}
	Since
	\begin{equation}\label{key}
		\begin{aligned}
			(\boldsymbol{\phi}_i-\psi_i^{*},\psi_i^{*}-\psi_i)	= (\boldsymbol{\phi}_i-\psi_i^{*},\psi_i^{*})-(\boldsymbol{\phi}_i-\psi_i^{*},\psi_i)
			=0
		\end{aligned}
	\end{equation}
	and
	\begin{equation}\label{key}
		(\psi_i^{*}-\psi_i,\psi_i^{*}-\psi_i) \geq 0,
	\end{equation}
	we have
	\begin{equation}\label{key}
		(\boldsymbol{\phi}_i-\psi_i,\boldsymbol{\phi}_i-\psi_i) \geq (\boldsymbol{\phi}_i-\psi_i^{*},\boldsymbol{\phi}_i-\psi_i^{*})
	\end{equation}
	Consequently, $ \psi_i^{*} $ represents the least squares approximation element in $ H(\Omega) $ for $ \boldsymbol{\phi}_i $. The aforementioned proofs of necessity and sufficiency culminate in the proof of Theorem 1. %\hfill $ \blacksquare $
\end{proof}

Based on Theorem 1, the optimization (\ref{equ:space_construction_optimization}) is transformed into the search for a collection of high-dimensional functions $ \{\psi_i^{*}\}_{i=1}^n $ that satisfy condition (\ref{equ:condition_in_theorem1}). In this context, the cubic B-spline method \cite{prautzsch2002bezier} is utilized to generate the continuous high-dimensional SBFs due to the following reasons:
\begin{itemize}
	\item [1)] In accordance with the \textit{continuity property} outlined in Ref. \cite{wei2022two}, the cubic B-spline surface maintains $ C^2 $ continuity, adhering to the design principle (a).
	\item [2)] Parametric B-spline surfaces can be transformed into functions of physical spatial coordinates, aligning with design principle (b).
	\item [3)] As outlined in the \textit{local modification property} in Ref. \cite{wei2022two}, B-spline functions fulfill the design principle (c).
	\item [4)] The B-spline features minimal support for a given degree, smoothness, and domain partition, rendering it inherently suitable for function approximation under sparse sensing. Moreover, cubic B-spline functions satisfy condition (\ref{equ:condition_in_theorem1}) when $ r=2 $.
\end{itemize}

The derivation of cubic B-spline functions proceeds as outlined below. Initially, the quasi-uniform knot vector is selected to guarantee that the designed B-spline surface exhibits continuous first- and second-order derivatives. Subsequently, the non-periodic knot vectors $ U $ and $ W $ of the cubic B-spline surface are defined as follows:
\begin{equation}\label{equ:knot_vector1}
	U=\left\{\begin{array}{l}
		\{0,0,0,0,1,1,1,1\}, N_1=4 \\
		\left\{0,0,0,0, \frac{1}{N_1-3}, \ldots, \frac{N_1-4}{N_1-3}, 1,1,1,1\right\}, N_1>4
	\end{array}\right.
\end{equation}
\begin{equation}\label{equ:knot_vector2}
	W=\left\{\begin{array}{l}
		\{0,0,0,0,1,1,1,1\}, N_2=4 \\
		\left\{0,0,0,0, \frac{1}{N_2-3}, \ldots, \frac{N_2-4}{N_2-3}, 1,1,1,1\right\}, N_2>4
	\end{array}\right.
\end{equation}
Based on the derivation of B-spline surfaces as described in \cite{wei2022two}, the spatial coordinates $ x $, $ y $, and $ z $ of the continuous SBF $ \psi_i(x,y) $ are formulated as functions of the parameters $ u $ and $ w $ and can be represented as:
\begin{equation}\label{equ:B_spline_surface_x}
	x = \sum_{j_1=1}^{N_1} V_{j_1,3}(u)x_{j_1}
\end{equation}
\begin{equation}\label{equ:B_spline_surface_y}
	y = \sum_{j_2=1}^{N_2} V_{j_2,3}(w)y_{j_2}
\end{equation}
\begin{equation}\label{equ:B_spline_surface_z}
	\begin{aligned}
		z = \psi_i(u,w) = \sum_{j_1=1}^{N_1} \sum_{j_2=1}^{N_2} V_{j_1, 3}(u) V_{j_2, 3}(w) \phi_i(x_{j_1}, y_{j_2})
	\end{aligned}
\end{equation}
where $ V_{j_1, 3}(u) $ and $ V_{j_2, 3}(w) $ represent the $ j_1 $th and $ j_2 $th 3-degree B-spline, respectively. The $ j $th k-degree B-spline is defined as described in \cite{wei2022two}:
\begin{equation}\label{key}
	\begin{aligned}
		V_{j, k}(u)= & \frac{u-u_{j}}{u_{j+k-u_{j}}} V_{j, k-1}(u) \\
		& +\frac{u_{j+k+1}-u}{u_{j+k+1}-u_{j+1}} V_{j+1, k-1}(u)
	\end{aligned}
\end{equation}
with the initial zero-degree B-spline defined as:
\begin{equation}\label{key}
	V_{j, 0}(u)= \begin{cases}1 & \text { if } u_{j} \leq u<u_{j+1}, j=1, \ldots, N_1 \\ 0 & \text { otherwise }\end{cases}
\end{equation}

Upon obtaining the spatial coordinates $ x $ and $ y $, the control parameters $ u $ and $ w $ can be determined using (\ref{equ:B_spline_surface_x}) and (\ref{equ:B_spline_surface_y}), respectively. The spatial coordinate $ z $ is then computed by substituting the derived parameters $ u $ and $ w $ into (\ref{equ:B_spline_surface_z}). This process yields the continuous SBFs $ {\psi_i(x,y)}_{i=1}^n $.

\subsection{Nonlinear learning of temporal dynamics}
During the online stage, assume the data snapshot $ \mathbf{T}_\text{s}[t] \in \mathcal{R}^{N_\text{s}\times 1}$ under sparse sensing at time $ t $ is obtained. According to (\ref{equ:sparse_KL}), the temporal coefficient vector $ \mathbf{a}_\text{s}[t] $ can be derived as:
\begin{equation}\label{equ:sparse_temporal_coefficient}
	\mathbf{a}_\text{s}[t] = \mathbf{\Phi}_\text{s}^\text{T} \mathbf{T}_\text{s}[t]
\end{equation}
Substituting (\ref{equ:sparse_temporal_coefficient}) into (\ref{equ:derivation_full_temporal_coefficient}), we have 
\begin{equation}\label{key}
	\mathbf{a}_\text{f}[t] = (\mathbf{\Phi}_\text{s}^\text{T} \mathbf{M} \mathbf{\Phi}_\text{f})^\dagger \mathbf{\Phi}_\text{s}^\text{T} \mathbf{T}_\text{s}[t]
\end{equation}
where $ \mathbf{a}_\text{f}[t] = [a_1(t),a_2(t),\ldots,a_n(t)]^\text{T} \in \mathcal{R}^{n\time 1}$. The temporal model can be formulated as follows:
\begin{equation}\label{key}
	\hat{\mathbf{a}}_\text{f}[t]  =F\left(\mathbf{a}_\text{f}[t-1],\ldots,\mathbf{a}_\text{f}[t-l_a],\mathbf{u}[t-1],\ldots,\mathbf{u}[t-l_u]\right)
\end{equation}
where $ F(\cdot) $ is a nonlinear function to capture the system temporal dynamics; $ \mathbf{u}[t] \in \mathcal{R}^{n_u}$ is the system input at time $ t $ with $ n_u $ denoting the number of inputs; $ l_a $ and $ l_b $ denote the output and input lags, respectively. Many nonlinear identification methods, such as support vector machines (SVM) \cite{shevade2000improvements}, radial basis function networks (RBFN) \cite{she2019battery}, recurrent neural networks (RNN) \cite{liu2022transferred}, etc., can be used to identify the temporal model $ F(\cdot) $.  

\begin{figure}[htbp] 
	\centering
	\includegraphics[width=0.65\textwidth]{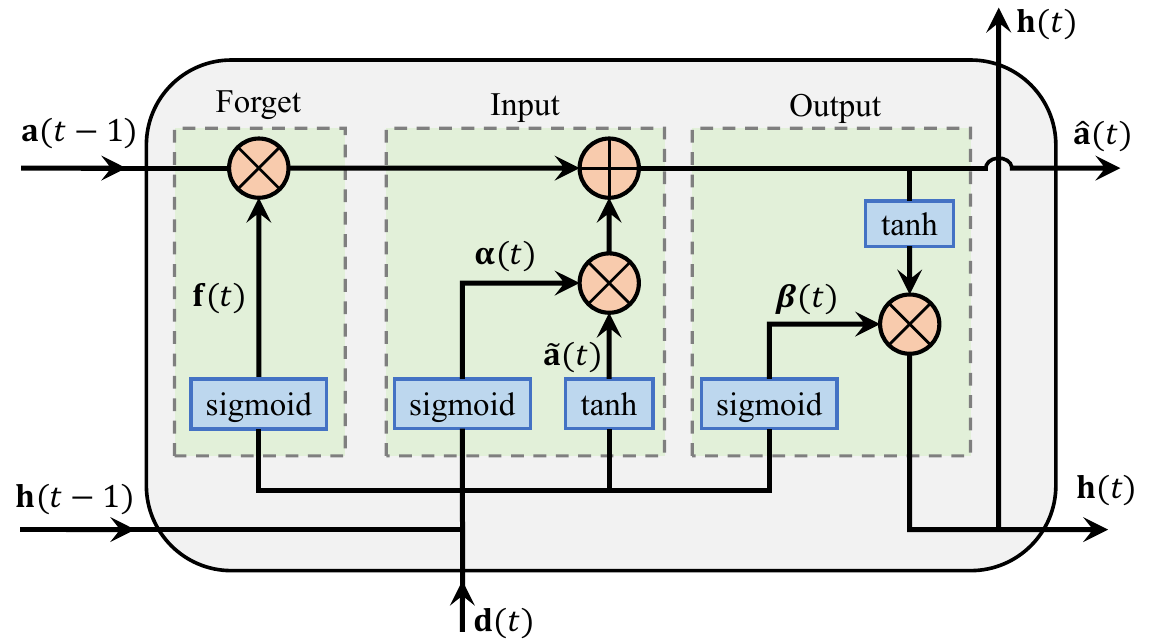}
	\caption{Architecture of a LSTM cell. The LSTM network utilized comprises the forget, input, and output layers. The symbols $\mathbf{h}(t)$, $\mathbf{f}(t)$, $\boldsymbol{\alpha}(t)$, $\boldsymbol{\beta}(t)$, and $\tilde{\mathbf{a}}(t)$ represent the hidden state, activation of the forget gate, activation of the input gate, activation of the output gate, and the potential output candidate, respectively.}
	\label{fig:LSTM}	
\end{figure}

Long short-term memory (LSTM) neural network is a type of recurrent neural network (RNN). The chain-type natural structure makes it naturally suitable for processing sequence-type data. Since the data in the temporal coefficients is highly correlated, the LSTM neural network is used to learn the temporal model as shown in Fig. \ref{fig:LSTM}. The input vector of the LSTM neural network is $ \mathbf{d}(t) = [a_1(t-1),\ldots,a_n(t-1),\mathbf{u}^\text{T}[t]]^\text{T} $. The output vector is $ \hat{\mathbf{a}}_\text{f}(t) = [\hat{a}_1(t),\hat{a}_2(t),\ldots,\hat{a}_n(t)]^\text{T} $. $ \mathbf{h}(t) $, $ \mathbf{f}(t) $, $ \boldsymbol{\alpha}(t) $, $ \boldsymbol{\beta}(t) $, and $ \tilde{\mathbf{a}}(t) $ are the hidden state, activation of forget gate, activation of input gate, activation of output gate, and potential output candidate, respectively. The detailed implementation of LSTM neural network can be referred to Refs. \cite{ojo2020neural,li2020battery,zhang2022data}.

\subsection{Space-time synthesis}
Following the continuous space construction and the nonlinear learning of temporal dynamics, the implementation of spatially continuous prediction is as follows:
\begin{equation}\label{equ:space-time_synthesis}
	\widehat{T}(x,y,t) = \sum_{i=1}^{n} \psi_i(x,y) \hat{a}_i(t)
\end{equation}
This process corresponds to the overall space estimation depicted in Fig. \ref{fig:conception}(c). It is important to note that the temporal coefficient function $ \hat{a}i(t) $ is acquired through learning from the sparse temporal coefficient $ \mathbf{a}\text{s} $. Consequently, the proposed method accomplishes complete space-time prediction even under sparse sensing conditions.

The following indexes are introduced for performance testing and comparisons:
\begin{itemize}
	\item[1)] Spatiotemporal absolute error:
	\begin{equation}\label{equ:STAE}
		\text{STAE}(x,y,t_i) = |\widehat{T}(x,y,t) - T(x,y,t_i)| 
	\end{equation}
	\item[2)] Spatial normalized absolute error:
	\begin{equation}\label{equ:SNAE}
		\text{SNAE}(t_i) = \frac{\int\int \text{STAE}(x,y,t_i)dxdy}{\int\int 1 dxdy}
	\end{equation}
	\item[3)] Root mean square error:
	\begin{equation}\label{equ:RMSE}
		\text{RMSE} = \left(\frac{\sum_{i=l_1+1}^{l_1+l_2}\int\int\text{STAE}(x,y,t_i)^2dxdy}{l_2\Delta t\int\int 1 dxdy}\right)^{1/2}
	\end{equation}
\end{itemize}
where $ T(x,y,t) $ and $ \widehat{T}(x,y,t) $ denote the measured and estimated outputs, respectively.


\bibliography{my_ref10}% common bib file

\end{document}